%% file: FAOC2020-SACM.tex
\documentclass[envcountsect,envcountsame]{fac}

\usepackage{FAOC2020-SACM}

\allowdisplaybreaks

\begin{document}

\title{Integration of Formal Proof into Unified Assurance Cases with
  \isacm}
\author{Simon Foster, Yakoub Nemouchi, Mario Gleirscher, Ran Wei, and Tim Kelly}
\correspond{Simon Foster, Email: \href{mailto:simon.foster@york.ac.uk}{simon.foster@york.ac.uk}.}
\makecorrespond
\maketitle

\begin{abstract}
  Assurance cases are often required to certify critical systems.  The
  use of formal methods in assurance can improve automation, increase
  confidence, and overcome errant reasoning.  However, assurance cases
  can never be fully formalised, as the use of formal methods is
  contingent on models that are validated by informal processes.
  Consequently, assurance techniques should support both formal and
  informal artifacts, with explicated inferential links between them.
  In this paper, we contribute a formal machine-checked interactive
  language, called \isacm, supporting the computer-assisted
  construction of assurance cases compliant with the OMG Structured
  Assurance Case Meta-Model. The use of \isacm guarantees
  well-formedness, consistency, and traceability of assurance cases,
  and allows a tight integration of formal and informal evidence of
  various provenance. In particular, Isabelle brings a diverse range
  of automated verification techniques that can provide evidence. To
  validate our approach, we present a substantial case study based on
  the Tokeneer secure entry system benchmark. We embed its functional
  specification into Isabelle, verify its security requirements, and
  form a modular security case in \isacm that combines the
  heterogeneous artifacts. We thus show that Isabelle is a suitable
  platform for critical systems assurance.
\end{abstract}

\begin{keywords}
 Assurance cases; Integrated Formal Methods; Proof Assistants; Safety
 Cases; Common Criteria
\end{keywords}

\section{Introduction}
\label{sec:intro}
\input{sec/intro}

\section{Preliminaries}
\label{sec:prelim}
\input{sec/prelim}

\section{Case Study: Tokeneer}
\label{sec:tokeneer}
\input{sec/tokeneer}

\section{\isacm}
\label{sec:isacm}
\input{sec/isacm}
\section{Modelling and Verification of Tokeneer}
\label{sec:model}
\input{sec/model}

\section{Mechanising the Tokeener Assurance Case}
\label{sec:tokassure}
\input{sec/new-tokassure}

\section{Related Work}
\label{subsec:relatedWork}
\input{sec/related}

\section{Conclusions}
\label{sec:conclusion}
\input{sec/concl}

\paragraph{Acknowledgments}

This work is funded by EPSRC projects CyPhyAssure\footnote{CyPhyAssure
  Project: \url{https://www.cs.york.ac.uk/circus/CyPhyAssure/}} (grant
reference EP/S001190/1) and RoboCalc (grant reference EP/M025756/1),
the German Science Foundation~(DFG; grant 381212925), and the Assuring
Autonomy International Programme~(AAIP; grant CSI:Cobot).

\bibliographystyle{alpha}
\bibliography{FAOC2020-SACM}

\end{document}

%% file: sec/intro.tex
\Acp{ac} are structured arguments, supported by evidence, intended to
demonstrate that a system meets its requirements, such as safety or
security, when applied in a particular operational
context~\cite{Wei2019-SACM,Kelly1998}.  They are recommended by
several international standards, such as ISO~26262 for automotive
applications.  An \ac{ac} consists of a hierarchical decomposition of
claims, through appropriate argumentation strategies, into further
claims, and eventually supporting evidence. Several \ac{ac} languages
exist, including the \gsn~\cite{Kelly1998}, Claims, Arguments, and
Evidence (CAE)~\cite{CAE}, and the \ac{sacm}~\cite{sacm,Wei2019-SACM},
which unifies several notations.

\Ac{ac} creation can be supported by model-based design, which
utilises architectural and behavioural models over which requirements
can be formulated~\cite{Diskin2018,Wei2019-SACM}.  However, \acp{ac}
can suffer from logical fallacies and inadequate
evidence~\cite{Greenwell2006}.
Moreover, although notations such as \gsn or CAE have their merits in
supporting the management of complex \acp{ac} via their hierarchical
decomposition and modular
representation~\cite{Gleirscher2017-ArguingHazardAnalysis}, improving
the comprehensibility and the maintenance of an assurance argument,
these notations have been criticised for ambiguities arising from the
existing
guidance~(e.g.~\cite{Denney2015-TowardsFormalBasis,Diskin2018}).

A proposed solution is formalisation in a machine-checked logic to
enable verification of consistency and
well-formedness~\cite{Rushby2013,Rushby2014}.  As confirmed by
avionics standard DO-178C, the evidence gathering process can also
benefit from the rigour of \acp{fm}~\cite{DO-333}. However, it is
also the case that (1) \acp{ac} are intended primarily for human
consumption, and (2) that formal models must be validated
informally~\cite{Habli2014}.  Consequently, \acp{ac} usually combine
informal and formal content, and so tools must support this.
Moreover, there is a need to integrate several
\acp{fm}~\cite{Paige1997FM-IntegratedFormalMethods}, potentially with
differing computational paradigms and levels of
abstraction~\cite{Hoare&98}.  So, for several reasons, it is paramount
to maintain change impact traceability across these heterogeneous
artifacts using \ac{fm}
integration~\cite{Gleirscher2018-NewOpportunitiesIntegrated}, going
beyond the possibilities of shallow hyperlinking of engineering data
as frequently seen in model-based engineering practice.

\begin{figure}
  \begin{center}
  \begin{tikzpicture}[dia,scale=.85,every node/.style={transform shape}]
    \node[itmooc] (arg) {Graphical\\Argument\\$\quad$};
    \node[lang] at (arg.south east) {e.g.~GSN/CAE tool};
    \node[tit] (argtit) at (arg.north west) {\S\ref{sec:prelim}};
    \node[itm,minimum width=3cm,text width=3cm] (isacm) at ($(arg)+(5,0)$) {Machine-checked\\
      SACM Model\\$\quad$};
    \node[tit] at (isacm.north west) {~~\S\ref{sec:isacm}};
    \node[tit] at (isacm.north west) {~~~~~~and~\S\ref{sec:tokassure}};
    \node[lang] at (isacm.south east) {Isabelle/SACM};
    \node[ann] (isacmann) at (isacm.south west) {heterogeneous and   
     \\ machine checked artifacts}; %
    \node[itm,minimum width=3cm,text width=3cm] (fm) at ($(isacm)+(5,1)$) {Integrated\\Formal Methods\\$\quad$};
    \node[lang] at (fm.south east) {Isabelle/UTP};
    \node[tit] at (fm.north west) {\S\ref{sec:model}};
    \node[itmooc,minimum width=3cm,text width=3cm] (infart) at ($(isacm)+(5,-1)$) {Other Resources\\$\quad$};
    \node[lang] at (infart.south east) {e.g.~ informal};

    \node[businessman,minimum size=1cm,align=center] (vereng) at
    ($(arg.west)+(-2,0)$) {Assurance\\Engineer};
    
    \draw[arrows={-latex},thick,tips=proper] 
    (vereng) edge[near start] node[lab] {uses} (arg)
    (vereng) edge[near start,bend right] node[lab] {uses}
    (isacm.south west)
    (arg) edge[bend left] node[lab,above] {translated to} (isacm)
    (isacm) edge node[lab] {feedback} (arg)
    (isacm) edge[] node {} (fm)
    (isacm) edge node[lab,above] {artifacts\\from} (infart)
    ;
  \end{tikzpicture}
\end{center}
\caption{Overview of our approach to integrative model-based assurance cases}
\label{fig:iFMAC}
\end{figure}
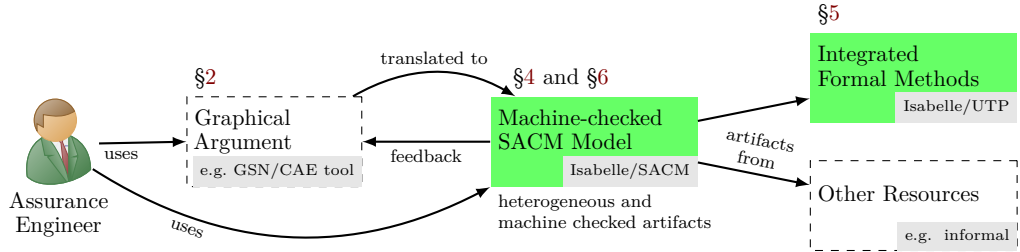

\paragraph{Vision}

Our vision, illustrated in Figure~\ref{fig:iFMAC}, is a unified
framework for machine-checked \acp{ac} with heterogeneous artifacts
and integrated \acp{fm}.  We envisage an assurance backend for a
variety of graphical assurance tools~\cite{Denney2018,Wei2019-SACM}
that utilise \ac{sacm}~\cite{sacm} as a unified interchange format,
and an array of FM tools provided by our verification platform,
\iutp~\cite{Foster2020-IsabelleUTP,Foster19a-IsabelleUTP,Foster16a}. Our
framework aims to improve existing processes by harnessing
formal verification to produce mathematically grounded \acp{ac} with
guarantees of consistency and adequacy of the evidence.  In the
context of safety regulation in domains such as intelligent
transportation, critical infrastructure, human-robot collaboration, or
medical devices, our framework can aid \ac{ac} evaluation
through machine-checking and automated
verification~\cite{Brucker2019-DOFCert,foster2020formal}.

\paragraph{Contributions}

A first step in this direction is made by the contributions of this paper, which are: (1) \isacm, an implementation of
SACM in Isabelle~\cite{Nipkow-Paulson-Wenzel:2002}, (2) a front-end for \isacm called interactive assurance language
(IAL), which is an interactive DSL for the definition of machine-checked SACM models, (3) a novel formalisation of
Tokeneer~\cite{Barnes2006-EngineeringTokeneerenclave} in \iutp, (4) the verification of the Tokeneer security
requirements\footnote{Supporting materials, including Isabelle theories, can be found on
  \href{http://www.cs.york.ac.uk/~simonf/FAOC2020}{our website}.}, and (5) the definition of a modular assurance case
capturing the lifecycle artifacts and the claims that Tokeneer meets its security requirements. Our Tokeneer assurance
case demonstrates how one can integrate formal artifacts, resulting from the work with \iutp (4), and informal
artifacts, such as the Tokeneer documentation.

Isabelle provides a sophisticated executable document model for
presenting a graph of hyperlinked formal artifacts, such as
definitions, theorems, and proofs. It provides automatic and incremental consistency checking, where updates to
artifacts trigger rechecking. Such capabilities can support efficient maintenance and evolution of model-based
\acp{ac}~\cite{Wei2019-SACM}.  Moreover, the document model allows both formal and informal
content~\cite{DBLP:journals/corr/abs-1811-10819}, and provides access to an array of automated proof
tools~\cite{Wenzel2007, DBLP:journals/corr/abs-1811-10819}.  Additionally, Brucker et
al.~\cite{DBLP:conf/mkm/BruckerACW18,Brucker2019-DOFDesign,Brucker2019-DOFCert} extend Isabelle with \idof, a framework
with a textual language for the embedding of ontologies and meta-models into the Isabelle document model, which we
harness to embed \ac{sacm}. For these reasons, we believe that Isabelle is an ideal platform both for assurance cases
and integration of formal methods.

Isabelle/UTP~\cite{Foster2020-IsabelleUTP} employs Unifying Theories
of Programming~\cite{Hoare&98} (UTP) to provide formal verification
facilities for a variety of languages, with paradigms as diverse as
concurrency~\cite{Foster17c}, real-time~\cite{Foster17b}, and hybrid
computation~\cite{Foster16b,Foster2020-dL}. Moreover, verification
techniques such as Hoare logic, weakest precondition calculus, and
refinement calculus are all available through a variety of proof
tactics. This makes \iutp an obvious choice for modelling and
verification of Tokeneer, and more generally as a platform for
integrated FMs based on unifying formal semantics.

We believe our novel mechanisation of the Tokeneer specification in a theorem prover is one of the most complete to date, with respect to
the original benchmark~\cite{TIS-FormalSpec}. The model includes 60 state variables, 38 top-level operations for user
entry, admin, and enrolment procedures, 30 invariants, and several hundred discharged invariant proof obligations. Where
possible, we are faithful to the benchmark, using the same names and structure for the system model. With our
mechanisation we are able to formally verify three security properties that could only be argued semi-formally in the
original documents~\cite{TIS-SecurityProperties}. Our work therefore demonstrates how automated proof tools have
advanced over the past fifteen years. We also highlight a few invariants missing from the original formal specification,
without which we could not verify Tokeneer.

This paper is an extension of a previous conference paper~\cite{Foster2019-iFM}. We develop a more elaborate modular
assurance case for Tokeneer (\S\ref{sec:tokeneer} and
\S\ref{sec:tokassure}), further develop our IAL~(\S\ref{sec:isacm}),
and formalise the Admin operations in the Tokeneer formal model and verify two additional security properties
(\S\ref{sec:model}). We also provide further implementation details and examples throughout, and in particular describe
our strategy for converting the Tokeneer Z schemas into Isabelle/UTP.

This article is organised as follows. In \S\ref{sec:prelim}, we outline preliminaries: SACM, Isabelle, and DOF.  In
\S\ref{sec:tokeneer} we describe the Tokeneer system. In \S\ref{sec:isacm}, we begin our contributions by describing
\isacm, which consists of the embedding of \ac{sacm} into \idof~(\S\ref{subsec:sacm-model}), and
IAL~(\S\ref{subsec:sacm-language}). In \S\ref{sec:model}, we model and verify Tokeneer in Isabelle/UTP. In
\S\ref{sec:tokassure}, we describe the mechanisation of the Tokeneer
AC in the ACME graphical assurance case tool
and
\isacm.  In \S\ref{subsec:relatedWork}, we indicate relationships to previous research.  After reflecting on our
approach in \S\ref{sec:disc}, we conclude in \S\ref{sec:conclusion}.

%% file: sec/prelim.tex
In this section, we provide background material on \acp{ac}, the SACM standard,
the Isabelle components, and \iutp, all required to follow our investigations in
\S\ref{sec:isacm}, \S\ref{sec:model}, and \S\ref{sec:tokassure}.

\subsection{Assurances Cases and SACM}

\begin{figure}
  \centering
  \includegraphics[width=.8\linewidth]{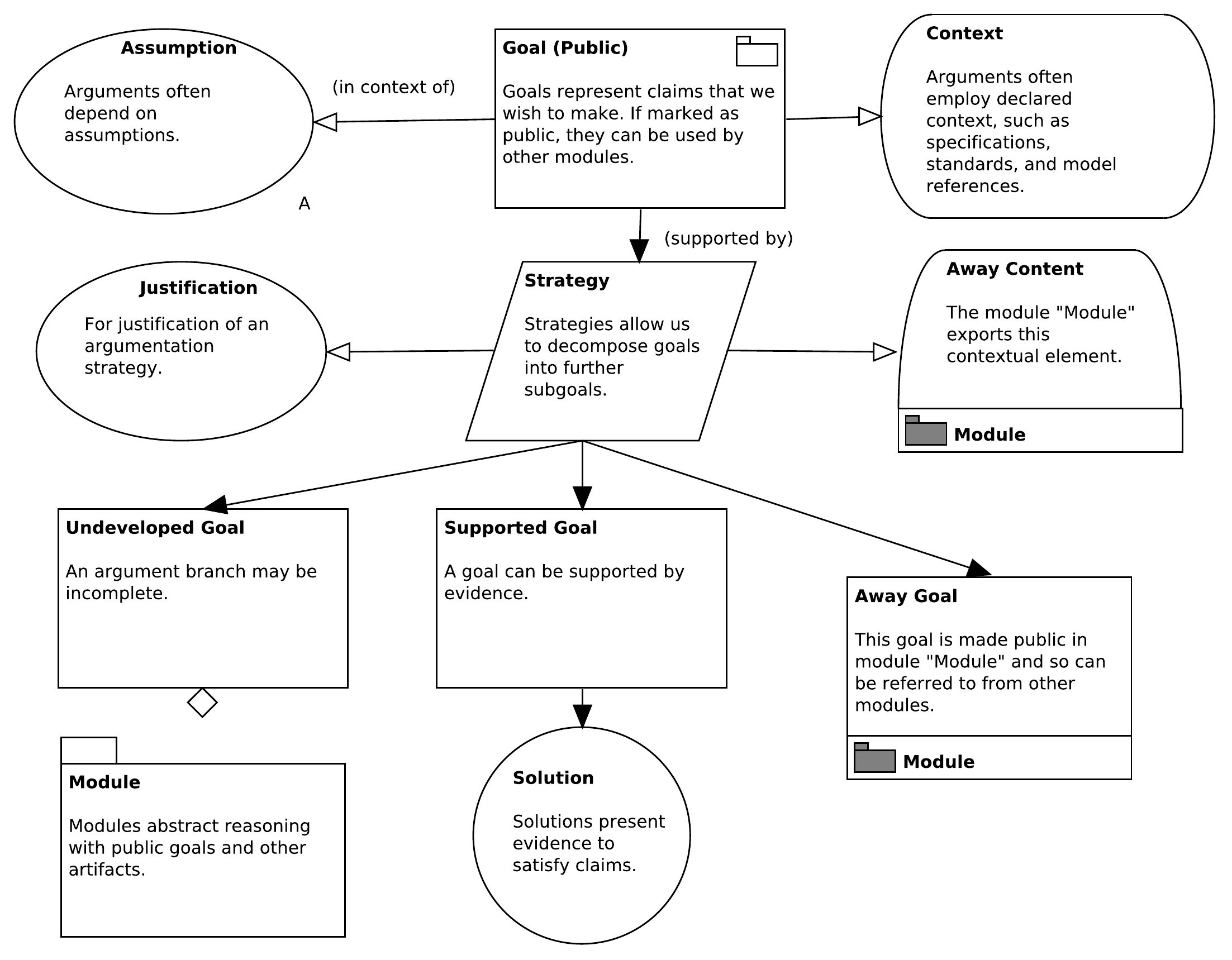}
  \caption{Goal Structuring Notation with Modular Extensions}
  \label{fig:gsn}
\end{figure}

Assurance cases are often presented using a graphical notation like GSN~\cite{Kelly1998} (Figure~\ref{fig:gsn}). In this
notation, claims are rectangles, which are linked with ``supported by'' arrows, strategies are parallelograms, and the
circles are evidence (``solutions''). The other shapes denote various types of context, which are linked to by the ``in
context of'' arrows. An argument in GSN proceeds from the most abstract claim down, through argumentation strategies and
further subgoals, until the claims can be directly supported by evidence. GSN also has a modular extension, where
arguments can be encapsulated in modules, and certain elements can be marked as public, meaning they are available to
other modules and can be cited using elements like ``away goals'' and ``away context''.

SACM is an OMG standard meta-model for ACs~\cite{Wei2019-SACM}. It unifies, extends, and refines several predecessor
notations, including GSN~\cite{Kelly1998} and CAE~\cite{CAE} (Claims, Arguments, and Evidence), and is intended as a
definitive reference model. SACM models three crucial concepts: arguments, artifacts, and terminology. An argument is a
set of claims, evidence citations, and inferential links between them. Artifacts represent evidence, such as system
models, techniques, results, activities, participants, and traceability links. Terminology fixes formal terms for use in
claims. Normally, claims are in natural language, but in SACM they can also contain structured expressions, which
allows integration of formal languages. Arguments, artifacts, and terminology can all be grouped into a number of
packages, which generalise GSN modules.

The argumentation meta-model of SACM is shown in Figure~\ref{fig:sacm-arg}. The base class is \textsf{ArgumentAsset},
which groups the argument assets, such as \textsf{Claim}s, \textsf{ArtifactReference}s, and
\textsf{AssertedRelationship}s (which are inferential links). Every asset may contain a \textsf{MultiLangString} that
provides a description, potentially in multiple natural and formal languages, and corresponds to contents of the shapes
in Figure~\ref{fig:gsn}.

\textsf{AssertedRelationship}s represent a relationship that exists between several assets. They can be of type
\textsf{AssertedContext}, which uses an artifact to define context; \textsf{AssertedEvidence}, which evidences a claim;
\textsf{AssertedInference} which describes explicit reasoning from premises to conclusion(s); or
\textsf{AssertedArtifactContext} which documents a dependency between the claims of two artifacts.

Both \textsf{Claim}s and \textsf{AssertedRelationship}s inherit from \textsf{Assertion}, because in SACM both claims and
inferential links are subject to argumentation and refutation. SACM allows six different classes of assertion, via the
attribute \textsf{assertionDeclaration}, including \textsf{axiomatic} (needing no further support), \textsf{assumed},
and \textsf{defeated}, where a claim is refuted. An \textsf{AssertedRelationship} can also be flagged as
\textsf{isCounter}, where counter evidence for a claim is presented.

\begin{figure}
  \centering
  \includegraphics[width=.9\linewidth]{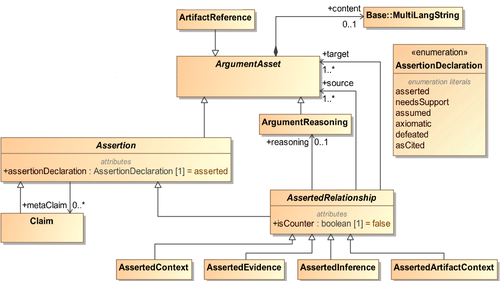}
  \caption{SACM Argumentation Meta-Model~\cite{sacm}}
  \label{fig:sacm-arg}
\end{figure}

For development of graphical assurance cases, we use an Eclipse-based tool called \textit{Assurance Case Management
  Environment} (ACME), from which we captured Figure~\ref{fig:gsn}. ACME supports the creation and management of
assurance cases using notations such as GSN (and in the future CAE), the abstract syntax of which is an extension of
Object Management Group's (OMG) Structured Assurance Case Metamodel (SACM), both are explained in detail
in~\cite{Wei2019-SACM}. ACME integrates a number of model management tools and frameworks, including Eclipse
Epsilon~\cite{kolovos2006eclipse}, Eclipse Hawk~\cite{barmpis2013hawk}, and Xtext~\cite{bettini2016implementing},
towards the management of fully model-based assurance cases. Using SACM's full potential and with the help of model
management frameworks, ACME currently supports 1) fine-grained traceability from an assurance case to its referenced
engineering models (defined in mainstream modelling technologies such
as, e.g.~UML) to the level of model elements; 2) traceability to
formal notations in Isabelle; 3) automated means to validate/verify traced engineering artifacts; 4) use and execution
of constrained natural language for model validation; 5) automated change impact analysis for assurance cases and their
engineering artifacts.

\subsection{Isabelle, Isar, and DOF}

\ihol is an interactive theorem prover for \ac{hol}~\cite{Nipkow-Paulson-Wenzel:2002}, based on the generic framework
Isar~\cite{Wenzel2007}. The former provides a functional specification language, and an array of automated proof
tools~\cite{Blanchette2011}. The latter has an interactive, extensible, and executable document
model~\cite{DBLP:journals/corr/abs-1811-10819}, which describes Isabelle theories. Plugins, such as \ihol, \idof, \iutp,
and \isacm have document models that contain conservative extensions to
Isar. %

\begin{wrapfigure}{r}{.25\linewidth}
    \includegraphics[width=\linewidth]{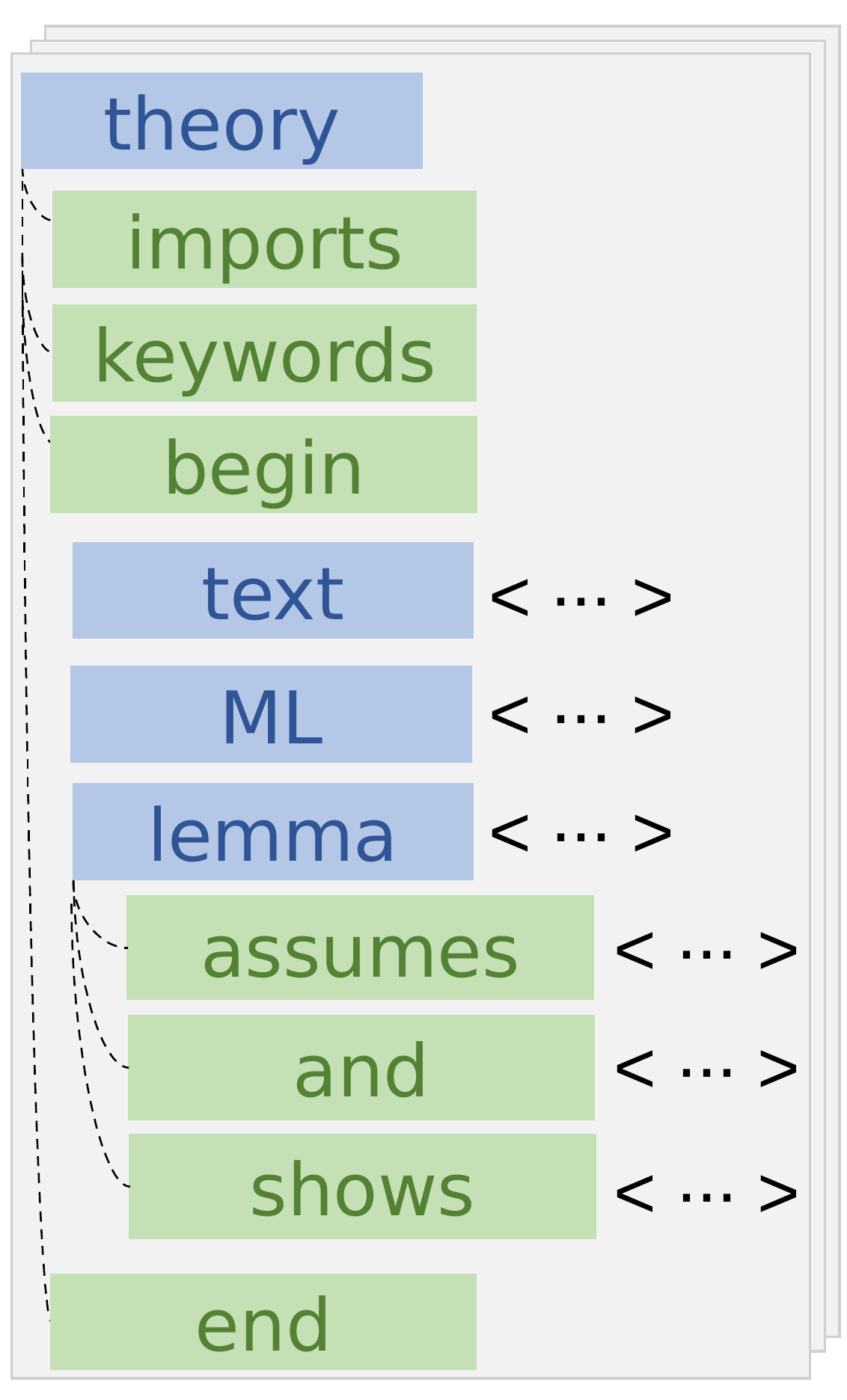}

    \caption{The Isabelle/Isar Document Model}
    \label{fig:doc-model}
\end{wrapfigure}

Figure~\ref{fig:doc-model} illustrates the document model. The first section for \emph{context definition} describes
\emph{imports} of existing theories, and \emph{keywords} which extend the concrete syntax. The second section is the
body enclosed between \emph{begin}-\emph{end}, which is a sequence of commands. The concrete syntax of commands consists
of (1) a pre-declared keyword (in \textcolor{Blue}{blue}), such as the command \textcolor{Blue}{\inlineisar+ML+}, (2) a
``semantics area'' enclosed between \inlineisar+\<open>...\<close>+, and (3) optional subkeywords (in
\textcolor{OliveGreen}{green}). Commands generate document elements. For example, the command
\textcolor{Blue}{\inlineisar+lemma+} creates a new theorem within the underlying theory context.  When a document is
edited by removal, addition, or alteration of elements, it is immediately executed and checked by Isabelle, with
feedback provided to the frontend. This includes consistency checks for the context and well-formedness checks for the
commands. Isabelle is therefore ideal for ACs, which have to be maintainable, well-formed, and consistent. In
\S\ref{subsec:sacm-language} we extend this document model with commands that define our assurance language.

Moreover, informal artifacts in Isabelle theories can be combined with formal artifacts using the command
\textcolor{Blue}{\inlineisar+text+} \inlineisar+\<open>...\<close>+. It is a processor for markup strings containing a
mixture of informal artifacts and hyperlinks to formal artifacts through \emph{antiquotations} of the form
\inlineisar+@{aqname ...}+. For example, \textcolor{Blue}{\inlineisar+text+} \inlineisar+\<open>The reflexivity theorem @{thm HOL.refl}\<close>+ mixes natural language with a hyperlink to the formal artifact \inlineisar+HOL.refl+ through
the antiquotation \inlineisar+@{thm HOL.refl}+. This is important since antiquotations are also checked by Isabelle as
follows: (1) whether the referenced artifact exists within the underlying theory context; (2) whether the type of the
referenced artifact matches the antiquotation's type.

A major foundation for our work is
Isabelle/\idof~\cite{DBLP:conf/mkm/BruckerACW18,Brucker2019-DOFDesign,Brucker2019-DOFCert}, an ontology framework for
Isabelle. \idof permits the description of ontologies using \ac{iosl}, a language to model \emph{document classes},
which extends the document model with new structures. We use the command \textcolor{Blue}{\inlineisar+doc_class+} from
\ac{iosl} to add new document classes for each of the SACM classes. Instances of DOF classes are not embedded into the
HOL logic as datatypes, but sit at the meta-logical level in the document model. This means they can refer to other
objects like theorems and definitions, can themselves be referenced using antiquotations, and carry an enriched version
of the corresponding Isabelle markup string. One of DOF's targets is formal development of certification documents,
making use of Isabelle's proof facilities~\cite{Brucker2019-DOFCert}. In this work, we take their vision forward with
our SACM-based assurance case framework.

\subsection{Isabelle/UTP}
\label{sec:iutp}

\iutp~\cite{Foster2020-IsabelleUTP} is a tool for developing formal semantics and verification tools based on Hoare and
He's Unifying Theories of Programming~\cite{Hoare&98}. \iutp contains a number of theories for reasoning about programs
built using different computational paradigms, such as concurrent and real-time programming. In this paper, we use the
core relational programming model to verify the functional specification of Tokeneer.

Variable mutation in \iutp is modelled algebraically using lenses~\cite{Foster07}. A variable of type $V$ in a state
space $S$ is denoted by a lens $x : V \lto S$, with two functions $\lget : S \to V$ and $\lput : S \to V \to S$, that
respectively query and update the value of the variable in a given state $s : S$. This allows us to treat variables as
semantic objects, rather than syntactic objects. We can check whether two lenses, $x$ and $y$, refer to disjoint regions
of the state space using independence $x \lindep y$. We can also check whether an expression $e$ depends on a particular
variable $x$ using unrestriction, which is written $x \unrest e$, and is a semantic encoding of variable freshness. For
example, if $x \lindep y$ (they are different variables), then $x \unrest (y + 1)$, since the valuation of $y + 1$ does
not depend on the value of $x$. Using these predicates, and the UTP relational program model~\cite{Hoare&98}, we can
express laws about assignments, such as commutativity:
$$(x := e \relsemi y := f) = (y := f \relsemi x := e), \quad x \lindep y, y \unrest e, x \unrest f$$
As we have previously shown~\cite{Foster2020-IsabelleUTP}, lenses effectively allow us to semantically characterise
variable sets as well (for example $a = \{x, y, z\}$) and thus framing properties. As a dual to unrestriction, we also
have the used-by predicate, $a \usedby e$, which states that $e$ uses only those variables mentioned in $a$.

Rather than using the Z notation~\cite{Spivey89}, we use a variant of Dijkstra's \ac{gcl}~\cite{Dijkstra75} encoded in
\iutp to specify the model's behaviour.  Our \ac{gcl} has the following syntax:
$$\prog ::= \ckey{skip} | \ckey{abort} | \prog \relsemi \prog | \expr \longrightarrow \prog | \prog \intchoice \prog | \var := \expr | \uframe{\var}{\prog} $$
Here, $\prog$ is a program, $\expr$ is an expression, and $\var$ is a variable. The language provides sequential
composition, guarded commands, non-deterministic choice\footnote{Technically, this is the same as $\lor$ in \iutp, the
  same as in Z.}, and assignment. We adopt a frame operator $\uframe{a}{P}$, which states that $P$ changes only
variables in the namespace $a$~\cite{Foster16a,Foster19a-IsabelleUTP}. The namespace is modelled by a lens
$a : S_1 \lto S_2$, which shows how to embed the inner state-space $S_1$ into the outer state-space $S_2$. This enables
modular reasoning about the \ac{tis} internal and real-world states, which is a further novelty of our work. We give both a
weakest precondition ($\ckey{wp}$) and weakest liberal precondition ($\ckey{wlp}$) semantics to our GCL. Technically,
each operator is denoted as a relational predicate in UTP, and the following laws are theorems of these
definitions~\cite{Hoare&98,Cavalcanti&06}.

\begin{theorem}[UTP Weakest Preconditions] $ $ \label{thm:utp-wp} \isalink{https://github.com/isabelle-utp/utp-main/blob/master/utp/utp_wp.thy} %

\vspace{-2ex}
  
\begin{center}
\begin{minipage}{.4\linewidth}
\centering\begin{align*}
  \ckey{skip} \uwp b &= b \\
  \ckey{abort} \uwp b &= false \\
  (P \relsemi Q) \uwp b &= P \uwp (Q \uwp b) \\
  (e \longrightarrow P) \uwp b &= (e \land P \uwp b)
\end{align*}
\end{minipage}\begin{minipage}{.4\linewidth}
\centering\begin{align*}
  (P \intchoice Q) \uwp b &= P \uwp b \lor Q \uwp b \\
  x := e \uwp b &= b[e/x] \\
  \uframe{a}{P} \uwp b &= (b \land (P \uwp true)_{\uparrow a}) & a \unrest b \\
  \uframe{a}{P} \uwp b &= (P \uwp b_{\downarrow a})_{\uparrow a} & a \usedby b
\end{align*}
\end{minipage}
\end{center}
\end{theorem}

\noindent With these equations, we can calculate the weakest precondition of any program composed of these operators. The
$\ckey{wlp}$ semantics is almost the same for each operator, except for the following equations:
$$\ckey{abort} \uwlp b = true \qquad (e \longrightarrow P) \uwlp b = (e \implies P \uwp b) \qquad (P \intchoice Q) \uwlp b = P \uwp b \land Q \uwp b$$ Most of the $\uwp$ and $\uwlp$ laws are standard~\cite{Dijkstra75}, the exception
being the laws for the frame operator. These make use of state space coercions~\cite{Foster20-LocalVars},
$P_{\uparrow a}$ and $P_{\downarrow a}$, which respectively grow and
shrink
the state space of $P$ using $a$. This, for example, means that the types of variables and quantifiers in $P$ are type
coerced. The first frame law has the proviso $a \unrest b$, meaning that the postcondition $b$ does not depend on any
variables in the frame $a$. Consequently, the weakest precondition is essentially $b$, but we also need to conjoin the
domain of $P$. Since $P$ operates on the inner state space, we need to grow its state space using coercions. The second
frame law, conversely, has that $b$ depends only on the variables in $a$. Consequently, the weakest precondition is
derived directly from $P$, but with suitable state space coercions applied.

We can use $\ckey{wlp}$ calculus to verify Hoare triple, using the following well-known
theorem~\cite{Foster2020-IsabelleUTP}:

\begin{theorem} $\hoaretriple{p}{Q}{r} \iff (p \implies Q \uwlp r)$ \isalink{https://github.com/isabelle-utp/utp-main/blob/105fbfee39c68177c730abb2ec47def5920d1e8b/utp/utp_wlp.thy\#L80}

\end{theorem}

\noindent In \iutp we have developed a tactic, \inlineisar+hoare_wlp_auto+, that utilises this theorem, calculates the
precondition using Theorem~\ref{thm:utp-wp}, and uses the UTP tactic \inlineisar+rel_auto+ relational calculus
tactic~\cite{Foster2020-IsabelleUTP} to try and discharge the resulting verification condition.

%% file: sec/tokeneer.tex
\begin{figure}
  \subcaptionbox{Overview of the physical infrastructure
    \label{fig:tis:arch}}{
    \includegraphics[width=7cm]{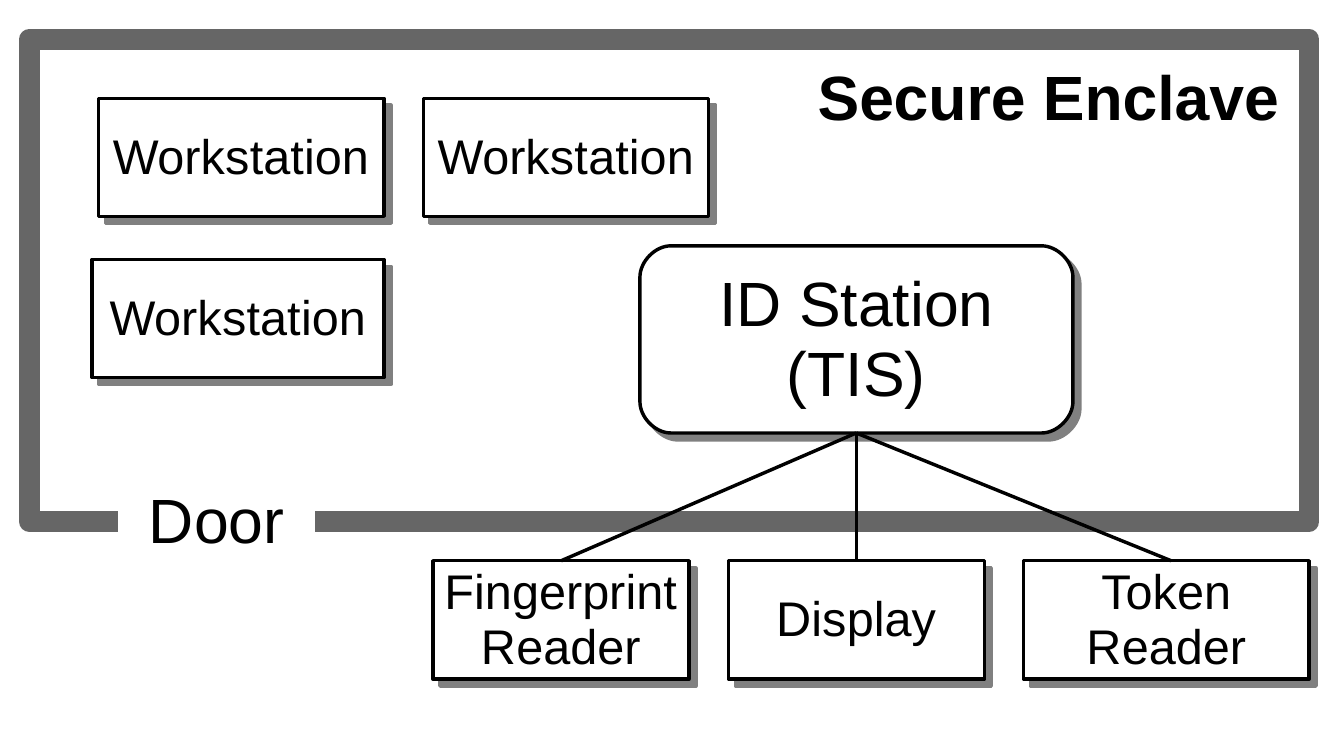}}
  \hfill
  \subcaptionbox{Fragment of the original assurance procedure
    \label{fig:tis:proc}}{\footnotesize\input{Figures/procedure}}
  \caption{The Tokeneer ID Station architecture and assurance
    procedure~(adapted from \cite{TIS-SummaryRep})
    \label{fig:tis}}
\end{figure}

To demonstrate our approach, we use the \acf{tis}\footnote{Project
  website: \url{https://www.adacore.com/tokeneer}} illustrated in
Figure~\ref{fig:tis}, a system that guards entry to a secure
enclave. The pioneering work on the \ac{tis} assurance was carried out
by Praxis High Integrity Systems and
SPRE~Inc.~\cite{Barnes2006-EngineeringTokeneerenclave}. Barnes et
al.~performed security analysis, definition of a security target,
formal functional specification using Z, refinement to a formal
design, implementation in SPARK, and verification of the security
properties against the Z specification~(Figure~\ref{fig:tis:proc}).

After independent assessment, \ac{cc} \ac{eal} 5 was achieved.
Therefore, Tokeneer can be seen as a successful example of using
\acp{fm} to assure a system against \ac{cc}. Though now more than
fifteen years old, it remains an important benchmark for \acp{fm} and
other assurance techniques.

As indicated in Figure~\ref{fig:tis:arch}, the physical infrastructure
consists of a door, fingerprint reader, display, and card (token)
reader. The main function is to check the credentials on a presented
token, read a fingerprint if necessary, and then either unlatch the
door, or deny entry. Entry is permitted when the token holds at least
three data items: (1) a user identity (ID) certificate, (2) a
privilege certificate, with a clearance level, and (3) an
identification and authentication (I\&A) certificate, which assigns a
fingerprint template. When the user first presents their token, the
three certificates are read and cross-checked. If the token is valid,
then a fingerprint is taken, which, if validated against the I\&A
certificate, allows the door to be unlocked once the token is
removed. An optional authorisation certificate is written upon
successful authentication, which allows the fingerprint check to be
skipped.

The \ac{tis} has a variety of other functions related to its
administration. Before use, a \ac{tis} must be enrolled, meaning it is
loaded with a public key chain and certificate, which are needed to
check token certificates. Moreover, the \ac{tis} stores audit data which
can be used to check previously occurred entries. The \ac{tis} therefore
also has a keyboard, floppy drive, and screen to configure
it. Administrators are granted access to these functions. The \ac{tis} also
has an alarm which will sound if the door is left open for too long.

The security of the \ac{tis} is assured by demonstrating six
\acp{sfr}~\cite{TIS-SecurityProperties}:
\label{props:TIS}
\begin{description}
\item[SFR1] \label{prop:SFR1} If the latch is unlocked, then \ac{tis} must
  possess either a User token or an Admin token. The User token must
  either have a valid authorisation certificate, or valid ID,
  Privilege, and I\&A Certificates, together with a template that
  allowed to successfully validate the User's fingerprint. Or, if the
  User token does not meet this, the Admin token must have a valid
  authorisation certificate, with the role ``guard''.
\item[SFR2] \label{prop:SFR2} If the latch is unlocked automatically by \ac{tis}, then the
  current time must be close to being within the allowed entry period
  defined for the User requesting access.
\item[SFR3] \label{prop:SFR3} An alarm will be raised whenever the door/latch is
  insecure.
\item[SFR4] \label{prop:SFR4} No audit data is lost without an audit alarm being raised.
\item[SFR5] \label{prop:SFR5} The presence of an audit record of one type will always be
  preceded by certain other audit records.
\item[SFR6] \label{prop:SFR6} The configuration data will be changed, or information
  written to the floppy, only if there is an Admin person logged on to
  the \ac{tis}.
\end{description}
\noindent Our objective is to (i) construct a machine-checked
assurance case that argues that the \ac{tis} fulfils the security
properties SFR1, part of SFR2, SFR3, and SFR6, and (ii) integrate
evidential artifacts from the mechanised model of the \ac{tis}
behaviour in \iutp into this assurance case.
For these SFRs, our approach re-enacts the green parts in
Figure~\ref{fig:tis:proc}.  Particularly, we focus on verifying the
functional formal specification against the security properties and on
checking well-formedness of the functional specification.

We envisage the modular assurance
case~\cite{Kelly1998,Denney2015-TowardsFormalBasis} for Tokeneer
illustrated in Figure~\ref{fig:tis-modular}. Here, we have modelled
the main documents produced during the development process as
assurance case modules using modular \ac{gsn}. The numbers correspond to the
document codes given in the Tokeneer archive\footnote{Tokeneer
  materials: \url{https://www.adacore.com/tokeneer/download}}. Each of
the package symbols represents a collection of claims, arguments, and
other lifecycle artifacts, for example
\textsf{40\_4\_Security\_Properties} provides formalisation of some of
the six SFRs. Certain artifacts are marked public, meaning they can be
used by other modules, and some are private. The arrows between the
modules indicate dependencies, for example the formal specification is
developed both in the context of the system requirements and the
security properties. In this paper, we focus on formalisation of
\textsf{41\_2\_Functional\_Specification}, and the argument that the
SFRs are satisfied in \textsf{TIS\_SFRs}. The assurance arguments and
artifacts will be embedded into Isabelle/SACM, which we develop in the
next section.

\begin{figure}
  \begin{center}
    \includegraphics[width=\linewidth]{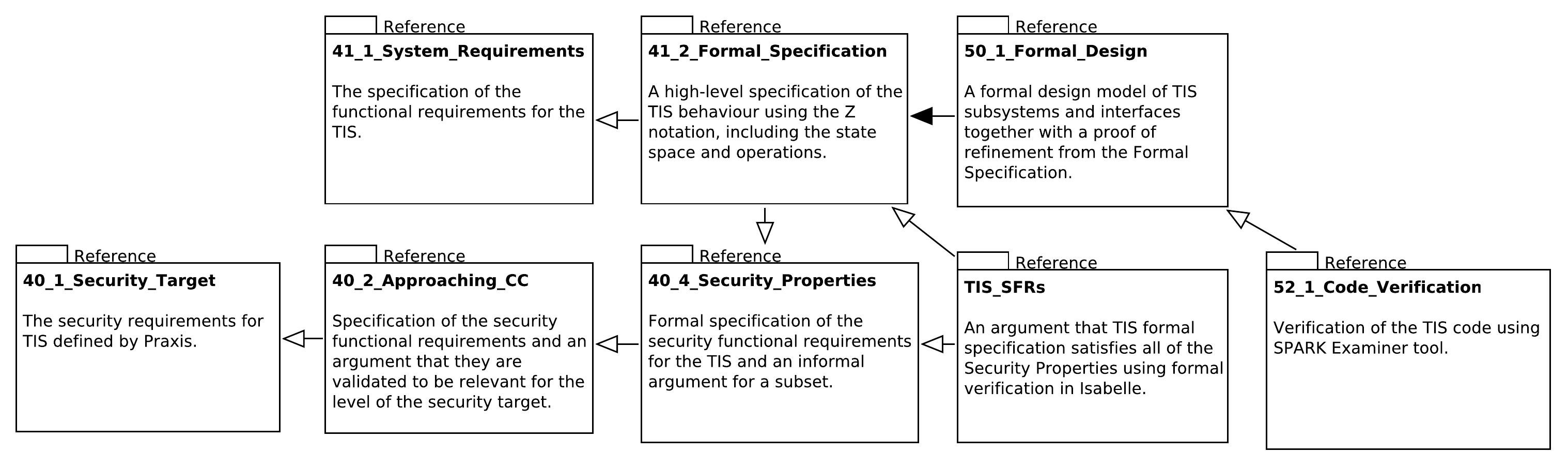}
  \end{center}
  \caption{TIS Modular Assurance Case Fragment}
  \label{fig:tis-modular}
\end{figure}

%% file: Figures/procedure.tex
\begin{tikzpicture}[dia,scale=.85,every node/.style={transform shape}]
  \node[itm] (prop) at (0,0) {Security\\Properties\\$\quad$};
  \node[lang] at (prop.south east) {natural language};
  \node[act] (proof1) at ($(prop)+(0,-2.5)$) {Proof of\\Security\\Properties};
  \node[itm] (fspec) at ($(proof1)+(3.5,2.5)$) {Formal\\Specification\\$\quad$};
  \node[lang] at (fspec.south east) {Z};
  \node[act] (proof2) at ($(fspec)+(3.5,0)$) {Proof of\\Formal\\Specification};
  \node[itmooc] (fdsgn) at ($(fspec)+(0,-2.5)$) {Formal\\Design\\$\quad$};
  \node[lang] at (fdsgn.south east) {Z};
  \node[actooc] (proof3) at ($(fdsgn)+(3.5,0)$) {Refinement\\Proof of Formal\\Design};
  
  \draw[arrows={-latex},thick,tips=proper] 
  (prop) edge node[align=center,left] {translated\\into Z for} (proof1)
  (fspec) edge[bend right=10] node[lab,above] {used in} (proof1)
  (proof1) edge[bend right=10] node[align=center,sloped] {satisfaction\\informally argued} (fspec)
  (fspec) edge[bend left] node[lab,below] {used in} (proof2)
  (proof2) edge node[lab,below] {verifies\\well-\\formedness} (fspec)
  (fspec) edge[bend right=10] node[lab,below] {used in} (proof3)
  (fdsgn) edge node[lab,below] {used\\in} (proof3)
  (proof3) edge[bend left] node[lab,below] {verifies refinement} (fdsgn)
  ;
  \draw[black!50,-latex,thick,decoration={zigzag,post=lineto,post
    length=.25cm},rounded corners=1pt]
  decorate {(fspec) -> (fdsgn)};
  \draw[black!50,-latex,thick,decoration={zigzag,post=lineto,post
    length=.25cm},rounded corners=1pt]
  decorate {(fdsgn) -> ($(fdsgn)+(0,-1.5)$)};
  \draw[black!50,decoration={zigzag},rounded corners=2pt] decorate
  {(-1.5,-4) -- (8,-4)};
  \node[align=center] at (0,-4.5) {System Test};
  \node[align=center] at (5,-4.5) {Informed Design, SPARK Implementation};
\end{tikzpicture}

%% file: sec/isacm.tex
In this section we encode SACM as a DOF ontology (\S\ref{subsec:sacm-model}), and use it to provide an interactive
machine-checked AC language (\S\ref{subsec:sacm-language}). Our embedding implements ACs as meta-logical entities in
Isabelle, that is, elements of the document model, rather than as formal elements embedded in the HOL logic, as this
would prevent the expression of informal reasoning and explanation. Therefore, antiquotations to formal artifacts can be
freely mixed with natural language and other informal artifacts.

\subsection{Modelling: Embedding SACM in Isabelle}
\label{subsec:sacm-model}
We embed the SACM meta-model in Isabelle using \ac{iosl}, and we focus
on modelling \inlineisar+ArgumentAsset+\footnote{We model all parts of
  SACM in \idof, but omit details for sake of brevity.} and its child
classes from Figure~\ref{fig:sacm-arg}, as these are the most relevant
classes for the TIS assurance argument that we develop in
\S\ref{sec:tokassure}.  The class \inlineisar+ArgumentAsset+ has the
following textual model:
\begin{isar}[numbers=none, backgroundcolor=\color{black!10}, frame=lines]
(*@\textcolor{Blue}{doc\_class}@*) ArgumentAsset = ArgumentationElement +
  content_assoc:: MultiLangString 
\end{isar}
Here, \textcolor{Blue}{\inlineisar+doc_class+} defines a new class, and automatically generates an antiquotation type,
\inlineisar+@{ArgumentAsset \<open>...\<close>}+, which can be used to refer to entities of this
type. \inlineisar+ArgumentationElement+ is a class which \inlineisar+ArgumentAsset+ inherits from, but is not discussed
further. \inlineisar+content_assoc+ models the content association in Figure~\ref{fig:sacm-arg}. To model
\inlineisar+MultiLang+\inlineisar+String+ in \isacm, we use \idof's markup string. Thus, the usage of antiquotations is
allowed for artifacts with the type \inlineisar+MultiLang+\inlineisar+String+.

\inlineisar+ArgumentAsset+ has three subclasses: (1) \inlineisar+Assertion+, which is a unified type for claims and
their relationships; (2) \inlineisar+ArgumentReasoning+, which is used to explicate the argumentation strategy being
employed; and (3) \inlineisar+ArtifactReference+, that evidences a claim with an artifact. Since \idof extends the \ihol
document model, we can use the latter's types, such as sets and enumerations (algebraic datatypes), in modelling SACM
classes, as shown below:
\begin{isar}[numbers=none, backgroundcolor=\color{black!10}, frame=lines]
(*@\textcolor{Blue}{datatype}@*) assertionDeclarations_t = 
  Asserted|Axiomatic|Defeated|Assumed|NeedsSupport|AsCited

(*@\textcolor{Blue}{doc\_class}@*) Assertion = ArgumentAsset +  
  assertionDeclaration::assertionDeclarations_t          

(*@\textcolor{Blue}{doc\_class}@*) Claim = Assertion +
  metaClaim::"Assertion set" <= "{}"

(*@\textcolor{Blue}{doc\_class}@*) ArgumentReasoning  = ArgumentAsset + 
  structure_assoc::"ArgumentPackage option"

(*@\textcolor{Blue}{doc\_class}@*) ArtifactReference = ArgumentAsset +
  referencedArtifactElement_assoc::"ArtifactElement set"
\end{isar}
Here, \textcolor{Blue}{\inlineisar+datatype+} defines an algebraic datatype, \inlineisar+assertionDeclarations_t+ is an
enumeration, \inlineisar+set+ is the set type, and \inlineisar+option+ is the optional type. Attribute
\inlineisar+assertionDeclaration+ is of type \inlineisar+assertionDeclarations_t+, which specifies the status of
instances of type \inlineisar+Assertion+. Examples of \inlineisar+Assertion+s in SACM are claims, justifications, and
both kinds of arrows in Figure~\ref{fig:gsn}. A \inlineisar+Claim+ is an assertion, extended with the
\inlineisar+metaClaim+ association. The attribute \inlineisar+structure_assoc+, in class \inlineisar+ArgumentReasoning+,
is an association to the class \inlineisar+ArgumentPackage+, which is not discussed here.  Finally, the attribute
\inlineisar+referencedArtifactElement_assoc+, from class \inlineisar+ArtifactReference+, is an association to
\inlineisar+ArtifactElement+s from the \inlineisar+ArtifactPackage+, allowing instances of type
\inlineisar+ArgumentAsset+ to be supported by evidential artifacts.

The class \inlineisar+Claim+ in Figure~\ref{fig:sacm-arg} inherits from the class \inlineisar+Assertion+ the attributes
\inlineisar+gid+, \inlineisar+content_assoc+, and \inlineisar+assertionDeclaration+ of type
\inlineisar+assertionDeclarations_t+. The other child class of \inlineisar+Assertion+ is
\inlineisar+AssertedRelationship+, as shown below.
\begin{isar}[numbers=none, backgroundcolor=\color{black!10}, frame=lines]
(*@\textcolor{Blue}{doc\_class}@*) AssertedRelationship = Assertion +
  isCounter::bool
  reasoning_assoc:: "ArgumentReasoning option"

(*@\textcolor{Blue}{doc\_class}@*) AssertedInference = AssertedRelationship +
  isCounter::bool <= False
  source::"Assertion set"
  target::"Assertion set"  
  
(*@\textcolor{Blue}{doc\_class}@*) AssertedEvidence = AssertedRelationship +
  isCounter::bool <= False
  source::"ArtifactAsset set"
  target::"Assertion set"  
\end{isar}
\inlineisar+AssertedRelationship+ models the relationships between instances of type \inlineisar+ArgumentAsset+, such as
the ``supported by'' and ``in context of'' arrows of Figure~\ref{fig:gsn}. \inlineisar+isCounter+ specifies whether the
target of the relation is supported or refuted by the source, and \inlineisar+reasoning_assoc+ is an association to
\inlineisar+ArgumentReasoning+, which models GSN strategies in \ac{sacm}. The attributes \inlineisar+source+ and
\inlineisar+target+, both of type \inlineisar+ArgumentAsset+, specify the source and target for the relation. Rather
than placing them directly in \inlineisar+AssertedRelationship+ we put them in the concrete subclasses, as this means
they can be specialised to enforce OCL constraints in the reference meta-model~\cite{sacm}. The various kinds of
relationship classes, such as \inlineisar+AssertedInference+ and \inlineisar+AssertedEvidence+, are then created as
subclasses. An \inlineisar+AssertedInference+ can only connect assertions, and an \inlineisar+AssertedEvidence+ can only
connect an evidential artifact to an assertion. These constraints are enforced by DOF when model instances are created.

\subsection{Interactive Assurance Language}
\label{subsec:sacm-language}
\acf{ial} is our assurance language with a concrete syntax consisting of various Isabelle commands that extend the
document model in Figure~\ref{fig:doc-model}. Each command performs a number of checks: (1) standard Isabelle checks
(\S\ref{sec:prelim}); (2) OCL-style constraints imposed on the attributes by SACM (provided by \idof); (3)
well-formedness checks against the meta-model, e.g. instances comply to the type restrictions imposed by the SACM
datatypes.

\begin{figure}
  \begin{center}
    \includegraphics[align=c,width=.45\linewidth]{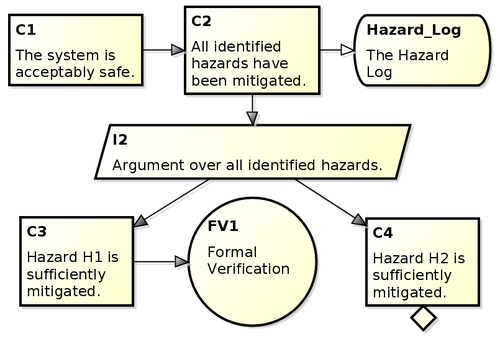}\qquad\includegraphics[align=c,width=.45\linewidth]{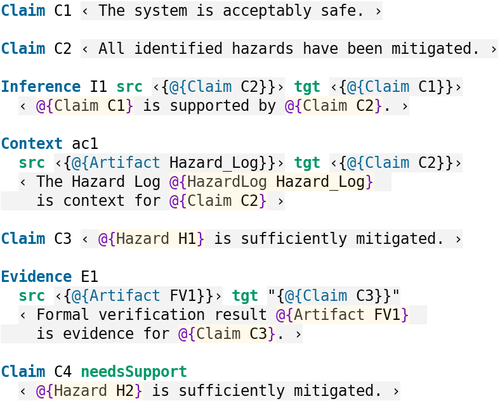}
  \end{center}

  \caption{Translation of GSN to IAL}
  \label{fig:GSN-IAL}
\end{figure}

IAL instantiates \textcolor{Blue}{\inlineisar+doc_class+}es from \S\ref{subsec:sacm-model} to create SACM model elements
in Isabelle, for example, the command \textcolor{Blue}{\inlineisar+Claim+} creates a model element of the class
\inlineisar+Claim+. Attributes and associations of a class have a concrete syntax represented by an Isabelle
(\textcolor{OliveGreen}{green}) subcommand. The grammar of the IAL commands for creating argumentation elements is shown
below.

\begin{isar}[numbers=none, backgroundcolor=\color{black!10}, frame=lines]
<AssertDecl>    := (*@\textcolor{OliveGreen}{asserted}@*) | (*@\textcolor{OliveGreen}{axiomatic}@*) | (*@\textcolor{OliveGreen}{assumed}@*) | (*@\textcolor{OliveGreen}{defeated}@*) | (*@\textcolor{OliveGreen}{needsSupport}@*)
<ClaimComm>     := (*@\textcolor{Blue}{Claim}@*) <gid> (*@\textcolor{OliveGreen}{isAbstract}@*)? (*@\textcolor{OliveGreen}{isCitation}@*)? ((*@\textcolor{OliveGreen}{metaClaims}@*) <gid>*)? <AssertDecl>? <Description>
<InferenceComm> := (*@\textcolor{Blue}{Inference}@*) <gid> <AssertDecl> ((*@\textcolor{OliveGreen}{src}@*) <gid>*) ((*@\textcolor{OliveGreen}{tgt}@*) <gid>*) <Description>
<ContextComm>   := (*@\textcolor{Blue}{Context}@*) <gid> <AssertDecl> ((*@\textcolor{OliveGreen}{src}@*) <gid>*) ((*@\textcolor{OliveGreen}{tgt}@*) <gid>*) <Description>
<EvidenceComm>  := (*@\textcolor{Blue}{Evidence}@*) <gid> <AssertDecl> ((*@\textcolor{OliveGreen}{src}@*) <gid>*) ((*@\textcolor{OliveGreen}{tgt}@*) <gid>*) <Description>
\end{isar}

\noindent\textcolor{Blue}{\inlineisar+Claim+} creates a model element of type \inlineisar+Claim+ 
with an identifier (\inlineisar+gid+), and description contained in a \inlineisar+MultiLangString+. The antiquotation
\inlineisar+@{Claim \<open><gid>\<close>}+ can be used to reference the created model element. The subcommands
\textcolor{OliveGreen}{\inlineisar+isAbstract+}, \textcolor{OliveGreen}{\inlineisar+isCitation+},
\textcolor{OliveGreen}{\inlineisar+metaClaims+}, and \textcolor{OliveGreen}{\inlineisar+<AssertDecl>+} are optional,
with default values being \inlineisar+False+, \inlineisar+False+, \inlineisar+{}+ and \inlineisar+asserted+,
respectively. The \textcolor{OliveGreen}{\inlineisar+metaClaims+} keyword allows us to link a claim to assertions about
this claim, such as the level of confidence in it. \textcolor{Blue}{\inlineisar+Inference+} creates an inference between
several model elements of type \inlineisar+ArgumentAsset+. It has subcommands \textcolor{OliveGreen}{\inlineisar+src+}
and \textcolor{OliveGreen}{\inlineisar+tgt+} that are both lists of antiquotations pointing to
\inlineisar+ArgumentAsset+s. The use of antiquotations to reference the instances ensures that Isabelle will do the
checks explained in~\S\ref{sec:prelim}. \textcolor{Blue}{\inlineisar+Context+} similarly asserts that an instance should
be treated as context for another, and \textcolor{Blue}{\inlineisar+Evidence+} associates evidence with a claim. Model
elements created by IAL are \emph{semi-formal}, since they can contain both informal content and references to machine
checked formal content.

With these commands, IAL can be used to represent a GSN diagram, as illustrated in Figure~\ref{fig:GSN-IAL}. The claims
\inlineisar+C1+--\inlineisar+C4+ are encoded using the \textcolor{Blue}{\inlineisar+Claim+} command. Claim
\inlineisar+C1+ is supported by \inlineisar+C2+ via the inference \inlineisar+I1+, which represents the ``supported-by''
arrow in the GSN at the left, and uses antiquotations to refer to the two claims. An artifact called
\inlineisar+Hazard_Log+ is introduced as context for \inlineisar+C2+ using the \textcolor{Blue}{\inlineisar+Context+}
command. A further evidence artifact \inlineisar+FV1+ is used to support claim \inlineisar+C3+, using the
\textcolor{Blue}{\inlineisar+Evidence+} command. The final claim \inlineisar+C4+ is left undeveloped, indicated by the
\textcolor{OliveGreen}{\inlineisar+needsSupport+} keyword.

\begin{figure}

  \begin{subfigure}{.5\linewidth}
    \centering
    \framebox{\includegraphics[width=.8\linewidth]{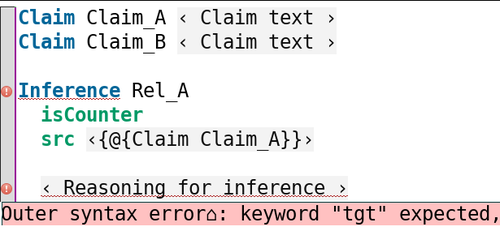}}

    \caption{Well-formedness}
    \label{sfig:well-formed}
  \end{subfigure}\begin{subfigure}{.5\linewidth}
    \centering
    \framebox{\includegraphics[width=.8\linewidth]{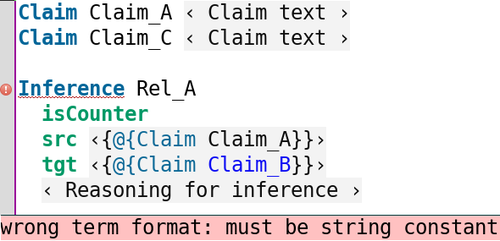}}

    \caption{Missing Elements}
    \label{sfig:missing}
  \end{subfigure}
  \begin{subfigure}{.5\linewidth}
    \centering
    \framebox{\includegraphics[width=.8\linewidth]{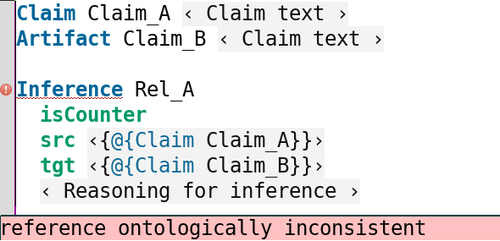}}

    \caption{Element Typing}
    \label{sfig:typing}
  \end{subfigure}\begin{subfigure}{.5\linewidth}
    \centering
    \framebox{\includegraphics[width=.8\linewidth]{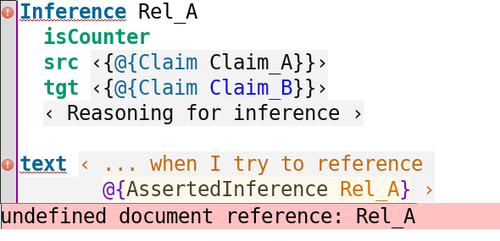}}

    \caption{Cascading Errors}
    \label{sfig:cascading}
  \end{subfigure}

  \caption{Interactive Nature of IAL Error Handling}
  
  \label{fig:ClaimA_DSL}

\end{figure}
Figure~\ref{fig:ClaimA_DSL} shows the interactive nature of IAL, and some of the error types. In
\eqref{sfig:well-formed}, the \textcolor{Blue}{\inlineisar+Inference+} command expects a source followed by target
element, but the latter is missing, and so IAL raises the error message at the bottom. The exclamation marks to the left
denote from where the error originates. In the jEdit interface, these errors are raised interactively whilst the user is
typing. Moreover, this kind of check ensures that the model elements produced conform to the SACM reference
meta-model.

In \eqref{sfig:missing}, the target is specified (\inlineisar+Claim_A+), but it refers to a claim that does not exist
(hence the blue colour), and so IAL again raises an error message. In \eqref{sfig:typing}, an element called
\inlineisar+Claim_A+ exists, but it is of the wrong type. \inlineisar+Claim_A+ is an artifact, which violates the OCL
constaints of the SACM standard~\cite{sacm}, and so DOF raises an ontological error. Finally, \eqref{sfig:cascading}
shows the cascading effect of errors: \inlineisar+Claim_B+ does not exist, the element \inlineisar+Rel_A+ fails to
process, and consequently any attempt to reference it will also fail. This kind of cascading can also be used to detect
proof failures following an update to a model and failed verification.

In addition to argumentation commands, we have also implemented several commands for creating different kinds of
artifacts.

\begin{isar}[numbers=none, backgroundcolor=\color{black!10}, frame=lines]
<ArtifactComm>    := (*@\textcolor{Blue}{Artifact}@*) <gid> ((*@\textcolor{OliveGreen}{version}@*) <string>)? ((*@\textcolor{OliveGreen}{date}@*) <string>)? <Description>
<RequirementComm> := (*@\textcolor{Blue}{Requirement}@*) <gid> ((*@\textcolor{OliveGreen}{version}@*) <string>)? ((*@\textcolor{OliveGreen}{date}@*) <string>)? <Description>
<ResourceComm>    := (*@\textcolor{Blue}{Resource}@*) <gid> ((*@\textcolor{OliveGreen}{location}@*) <URI>) <Description>
<ActivityComm>    := (*@\textcolor{Blue}{Activity}@*) <gid> ((*@\textcolor{OliveGreen}{startTime}@*) <string>)? ((*@\textcolor{OliveGreen}{endTime}@*) <string>)? <Description>
<EventComm>       := (*@\textcolor{Blue}{Event}@*) <gid> ((*@\textcolor{OliveGreen}{occurence}@*) <string>)? <Description>
<ParticipantComm> := (*@\textcolor{Blue}{Participant}@*) <gid> <Description>
<TechniqueComm>   := (*@\textcolor{Blue}{Technique}@*) <gid> <Description>
<ArtifactRelComm> := (*@\textcolor{Blue}{ArtifactRelation}@*) <gid> (*@\textcolor{OliveGreen}{src}@*) <gid>* (*@\textcolor{OliveGreen}{tgt}@*) <gid>* <Description>
\end{isar}

With the exception of \textcolor{Blue}{\inlineisar+Requirement+}, these artifact classes are adopted from the SACM
standard~\cite{Wei2019-SACM}. They allow us to model the various artifacts created during the development and assurance
lifecycle, and the relationships between them, for the purposes of traceability. The
\textcolor{Blue}{\inlineisar+Artifact+} command represents a unit of data produced during the lifecycle, such as a
specification or verification results. It can be annotated with a version, and the creation date. The
\textcolor{Blue}{\inlineisar+Requirement+} command can be used represent requirements, a specialised form of
\textcolor{Blue}{\inlineisar+Artifact+}. The \textcolor{Blue}{\inlineisar+Resource+} command can be used to model a link
to an external resource, such as a standard or code base, which is uniquely represented by a URI. The
\textcolor{Blue}{\inlineisar+Activity+} command models an activity or process, with a start time and end time, and
\textcolor{Blue}{\inlineisar+Event+} similarly represents a timed and dated event. A
\textcolor{Blue}{\inlineisar+Participant+} models an actor that takes part in the lifecycle, such as a developer, and
\textcolor{Blue}{\inlineisar+Technique+} models a technique, such as a modelling language or formal method, that is
applied in the creation of artifacts. Finally, \textcolor{Blue}{\inlineisar+ArtifactRelation+} allows us to relate two
artifacts.

An example using the artifact commands is shown in Figure~\ref{fig:IAL-Artifacts}, which further elaborates the
verification result in Figure~\ref{fig:GSN-IAL}. The formal verification result \inlineisar+FV1+ is an artifact, with
version 1, that points to the Isabelle theorem \inlineisar+vc1+. The result was created during a verification activity,
\inlineisar+VACT1+, as shown using the artifact relation \inlineisar+AR1+. Isabelle was used to perform the proof, which
is modelled using a \textcolor{Blue}{\inlineisar+Resource+} that links to the Isabelle website. The verification
activity was led by a proof engineer, Anne Other, who is modelled as a \textcolor{Blue}{\inlineisar+Participant+}, and
linked to the verification activity by a further artifact relation. The specific technique used for the proof was the
Isabelle simplifier, which is modelled as a \textcolor{Blue}{\inlineisar+Technique+}, and contains a link to the proof
method $simp$.

\begin{figure}
  \begin{center}
    \includegraphics[align=c,width=.75\linewidth]{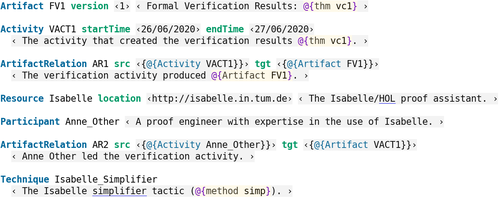}
  \end{center}

  \caption{Artifact Traceability in IAL}
  \label{fig:IAL-Artifacts}
\end{figure}

We have now developed \isacm and our IAL. In the next section, we consider the modelling verification of the Tokeneer
system.

%% file: sec/model.tex
In this section we present a novel mechanisation of Tokeneer in Isabelle/UTP~\cite{Foster16a,Foster19a-IsabelleUTP} to
provide evidence for the AC. This model encodes the formal functional specification (\textsf{41\_2}) in the modular
assurance case in Figure~\ref{fig:tis-modular}. In~\cite{TIS-SecurityProperties}, the satisfaction of the SFRs are
argued semi-formally using the functional specification, but here we provide a formal proof. We focus on the
verification of three of the requirements: SFR1 (the most challenging of the six), SFR3, and SFR6, and describe the
necessary model elements.

\subsection{Modelling and Mechanisation}

The TIS functional specification~\cite{TIS-FormalSpec} describes an elaborate state space and a collection of relational
operations. The state is bipartite, consisting of (1) the digital state of the TIS and (2) the monitored and controlled
variables shared with the real world. The TIS monitors the time, enclave door, fingerprint reader, token reader, and
several peripherals. It controls the door latch, an alarm, a display, and a screen.

The specification describes a complex state transition system, with around 50 operations for enrolling the station,
performing various administrative operations, such as archiving log files and updating the configuration file, and the
user entry operations. The main user entry operations are illustrated in Figure~\ref{fig:tis-states}
(cf. \cite[page~43]{TIS-FormalSpec}), where each transition corresponds to an operation. Following enrolment, the TIS
becomes \textsf{quiescent} (awaiting interaction). \textsf{ReadUserToken} triggers if the token is presented, and reads
its contents. Assuming a valid token, the TIS determines whether a fingerprint is necessary, and then triggers either
\textsf{BioCheckRequired} or \textsf{BioCheckNotRequired}. If required, the TIS then reads a fingerprint
(\textsf{ReadFingerOK}), validates it (\textsf{ValidateFingerOK}), and finally writes an authorisation certificate to
the token (\textsf{WriteUserTokenOK}). If the access credentials are available (\textsf{waitingEntry}), then a final
check is performed (\textsf{EntryOK}), and once the user removes their token (\textsf{waitingRemoveTokenSuccess}), the
door is unlocked (\textsf{UnlockDoor}).

\begin{figure}

  \centering\includegraphics[width=\linewidth]{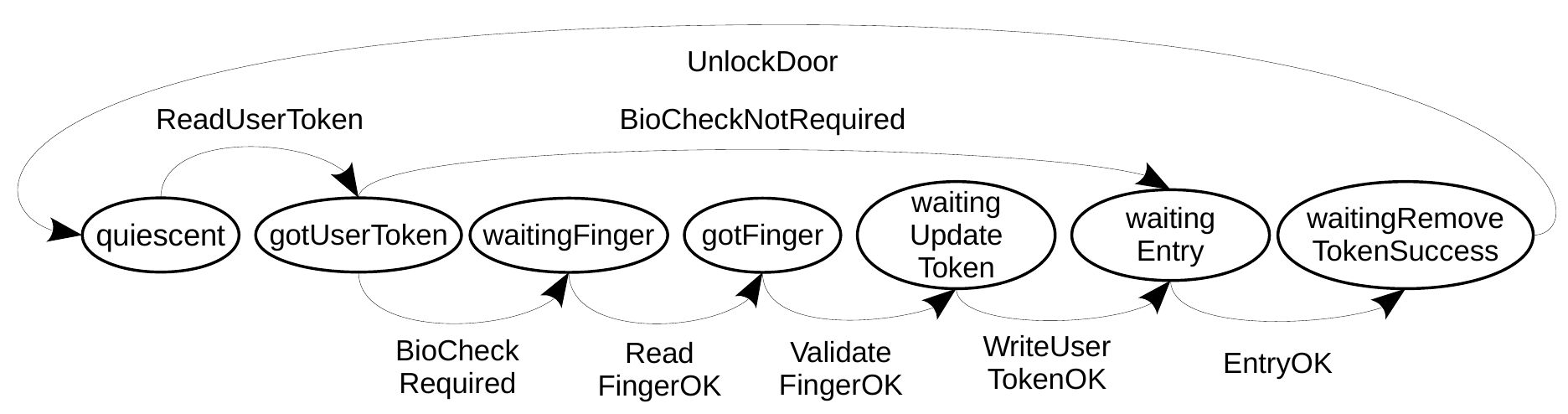}

  \caption{TIS Main States}
  \label{fig:tis-states}
  
\end{figure}

We mechanise the TIS using hierarchical state space types, with invariants adapted from the Z
specification~\cite{TIS-FormalSpec}. We define the operations using \ac{gcl}~\cite{Dijkstra75} rather than the Z schemas
directly, to enable syntax-directed reasoning. The syntax of Z~\cite{Spivey89}, though maximally flexible, does not
easily lend itself to such reasoning, since every operation schema contains a set of conjoined predicates which must be
considered in turn. In contrast, as illustrated in \S\ref{sec:iutp}, it is straightforward to calculate the weakest
precondition of a GCL program. Within these constraints, we have endeavoured to remain faithful to the function
specification by representing each of the state and operation schemas and using the same naming and overall
structure. The use of GCL means that the model is also much closer to a program, and consequently refinement to code
should be straightforward. Moreover, since GCL has a denotational semantics in UTP's relational
calculus~\cite{Hoare&98}, it may be possible to prove equivalence with the corresponding Z operations.

\subsection{State Space}

We first describe the state space of the TIS state machine:
\begin{definition}[TIS types and state space] \isalink{https://github.com/isabelle-utp/utp-main/blob/7abe4b02af634ee70503afc39fab60f46a8cf954/casestudies/Tokeneer/Tokeneer.thy\#L202}
\begin{align*}
\textit{LATCH} &\defs \mv{unlocked} | \mv{locked} \\
\textit{DOOR} &\defs \mv{open} | \mv{closed}  \\
\textit{TOKENTRY} &\defs \mv{noT} | \mv{badT} | \mv{goodT}~\textit{Token} \\
\textit{PRESENCE} &\defs \mv{present} | \mv{absent} \\
\textit{PRIVELEGE} &\defs \mv{userOnly} | \mv{guard} | \mv{securityOfficer} | \mv{auditManager} \\
\textit{ADMINOP} &\defs \mv{archiveLog} | \mv{updateConfigData} | \mv{overrideLock} | \mv{shutdownOp} \\
\textit{FLOPPY} &\defs \mv{noFloppy} | \mv{emptyFloppy} | \mv{badFloppy} | \mv{configFile}~(configFileOf: Config) | \cdots \\
  \textit{Config} &\defs \left[
\begin{array}{l}
  alarmSilentDuration : TIME, latchUnlockDuration : TIME, \\
  tokenRemovalDutation : TIME, enclaveClearance : CLEARANCE, \cdots
\end{array} \right] \\
IDStation &\defs \left[
\begin{array}{l}
  currentUserToken : TOKENTRY, currentTime : TIME, \\ 
  userTokenPresence : PRESENCE, status : STATUS, \\ 
  enclaveStatus : \textit{ENCLAVESTATUS}, currentDisplay : \textit{DISPLAYMESSAGE},  \\
  issuerKey : \textit{USER} \pfun \textit{KEYPART}, rolePresent : \textit{PRIVELEGE}~option, \\
  availableOps : \textit{ADMINOP}~set, currentAdminOp : \textit{ADMINOP}~option, \\
  ifloppy : FLOPPY, config : \textit{Config}, \cdots
\end{array}
\right] \\[.5ex]
Controlled &\defs \left[
\begin{array}{l}
  latch : LATCH, alarm : ALARM, display : DISPLAYMESSAGE, screen : Screen
\end{array}
\right] \\[.5ex]
Monitored &\defs \left[
\begin{array}{l}
  now : TIME, finger : \textit{FINGERPRINTTRY}, \\ 
  userToken : TOKENTRY, floppy : FLOPPY, keyboard : KEYBOARD
\end{array} 
\right] \\[.5ex]
RealWorld &\defs \left[ mon : Monitored, ctrl : Controlled \right] \\[.5ex]
SystemState &\defs [ rw : RealWorld, tis : IDStation ]
\end{align*}
\end{definition}
\noindent A collection of algebraic data types characterise the state of system elements, including the door, latch, and
token. The type $\alpha\,\textit{option}$ represents an optional value that can be either undefined, $\textit{None}$,
or defined, $\textit{Some}~x$ for $x : \alpha$. The function $the : \alpha\,\textit{option} \to \alpha$ allows us to
extract the value from a defined value. We define state types for the TIS state, controlled variables, monitored
variables, real-world, and the entire system, respectively. The controlled variables include the physical latch, the
alarm, the display, and the screen. The monitored variables correspond to time ($now$), the door ($door$), the
fingerprint reader ($finger$), the tokens, and the peripherals. \textit{RealWorld} combines the physical variables, and
\textit{SystemState} composes the physical world ($rw$) and the TIS ($tis$).

Variable \textit{currentUserToken} represents the last token presented to the TIS, and \textit{userTokenPresence}
indicates whether a token is currently present. The variable \textit{status} is used to record the state the TIS is in,
and can take the values indicated in the state bubbles of Figure~\ref{fig:tis-states}. Variable \textit{issuerKey} is a
partial function representing the public key chain, which is needed to authorise user entry. Variables
\textit{rolePresent}, \textit{availableOps}, and \textit{currentAdminOp} are used to represent the presence of an Admin,
the available operations for this Admin, and the current operation being executed.

In addition to the state types, we also encode a number of predicates that represent the invariants of seven Z state
schemas. These effectively encode low-level well-formedness constraints for the types; the higher level invariants are
considered in \S\ref{sec:formal-verify}. The predicate representing the invariants associated with the
Admin variables is shown below.

\begin{definition}[Administrator Invariants] \label{def:admin-invar} \isalink{https://github.com/isabelle-utp/utp-main/blob/7abe4b02af634ee70503afc39fab60f46a8cf954/casestudies/Tokeneer/Tokeneer.thy\#L391}
$$
Admin \defs \left(
\begin{array}{l}
  (\textit{rolePresent} \neq None \implies the(rolePresent) \in \{\mv{guard}, \mv{auditManager}, \mv{securityOfficer}\}) \\
  \land (\textit{rolePresent} = None \implies \textit{availableOps} = \{\}) \\
  \land (\textit{rolePresent} = Some(\mv{guard}) \implies \textit{availableOps} = \{\mv{overrideLock}\}) \\
  \land (\textit{rolePresent} = Some(\mv{auditManager}) \implies \textit{availableOps} = \{\mv{archiveLog}\}) \\
  \land (\textit{rolePresent} = Some(\mv{securityOfficer}) \implies \textit{availableOps} = \{\mv{updateConfigData}, \mv{shutdownOp}\}) \\
  \land (\textit{currentAdminOp} \neq None \implies the(\textit{currentAdminOp}) \in \textit{availableOps} \land \textit{rolePresent} \neq None)
\end{array} \right)
$$
\end{definition}

\noindent This predicate closely corresponds to the $Admin$ schema in the functional
specification~\cite[page~22]{TIS-FormalSpec}. It states, firstly, that if a role is present, it must be one of the three
Admin roles. Conversely, if no roles are present then no Admin operations are available to be executed. The next three
implications assign possible operations to the given Admin roles. The final predicate states that if an Admin operation
is being executed, then it must be one of the available operations and there must be a role present. We collect the seven
well-formedness predicates in \textit{TIS-wf}, as defined below.

\begin{definition}[Well-formedness Properties] \isalink{https://github.com/isabelle-utp/utp-main/blob/7abe4b02af634ee70503afc39fab60f46a8cf954/casestudies/Tokeneer/Tokeneer.thy\#L587}
  $$\textit{TIS-wf} \defs (DoorLatchAlarm \land Floppy \land KeyStore \land Admin \land Config \land AdminToken \land UserToken)$$
\end{definition}

\noindent The verification of the TIS SFRs depends on these state predicates being invariant for all the operations.

\subsection{Operations}

We now specify a selection of the operations over \textit{IDStation}\footnote{The TIS operations have been mechanised
  using the same names as in~\cite{TIS-FormalSpec}.}:
\begin{definition}[User Entry Operations] \label{def:tis-ops} \isalink{https://github.com/isabelle-utp/utp-main/blob/7abe4b02af634ee70503afc39fab60f46a8cf954/casestudies/Tokeneer/Tokeneer.thy\#L1391}
\begin{align*}
  ReadUserToken &\defs
  \begin{array}{l}
  \left(
  \begin{array}{l}
    enclaveStatus \in \left\{\begin{array}{l} \mv{enclaveQuiescent}, \\ \mv{waitingRemoveAdminTokenFail} \end{array}
    \right\}  \\
    \land status = \mv{quiescent} \land userTokenPresence = \mv{present}
  \end{array}\right) \\
  \longrightarrow currentDisplay := \mv{wait} \relsemi status := \mv{gotUserToken}
  \end{array} \\
  BioCheckRequired &\defs 
  \begin{array}{l}
  \left(\begin{array}{l}
    status = \mv{gotUserToken} \land userTokenPresence = \mv{present} \\
    \land UserTokenOK \land (\neg UserTokenWithOKAuthCert)
  \end{array}\right) \\[1ex]
  \longrightarrow status := \mv{waitingFinger}\!\relsemi currentDisplay := \mv{insertFinger}
  \end{array} \\
  ReadFingerOK &\defs
  \begin{array}{l}
  \left(\begin{array}{l}
    status = \mv{waitingFinger} \land fingerPresence = \mv{present} \\
    \land userTokenPresence = \mv{present}
  \end{array}\right) \\
  \longrightarrow status := \mv{gotFinger} \relsemi currentDisplay := \mv{wait}
  \end{array} \\
  UnlockDoor &\defs \left(
   \begin{array}{l}
     latchTimeout := currentTime + latchUnlockDuration \relsemi \\
     alarmTimeout := currentTime + latchUnlockDuration + alarmSilentDuration \relsemi \\
     currentLatch := \mv{unlocked} \relsemi doorAlarm := \mv{silent}
   \end{array} \right) \\
  UnlockDoorOK &\defs
  \begin{array}{l}
  \left(\begin{array}{l}
    status = \mv{waitingRemoveTokenSuccess} \\
    \land userTokenPresence = \mv{absent}
  \end{array}\right) \\
  \longrightarrow \begin{array}{l} UnlockDoor \relsemi status := \mv{quiescent} \relsemi \\ currentDisplay := \mv{doorUnlocked} \end{array}
  \end{array}
\end{align*}
\end{definition}
\noindent Each operation is guarded by execution conditions and consist of several assignments. \textit{BioCheckRequired}
requires that the current state is $\mv{gotUserToken}$, the user token is $\mv{present}$, and sufficient for entry
($UserTokenOK$), but there is no authorisation certificate ($\neg UserTokenWithOKAuthCert$).  The latter two predicates
essentially require that (1) the three certificates can be verified against the public key store, and (2) additionally
there is a valid authorisation certificate present. We give the definition of $UserTokenOK$ below. \isalink{https://github.com/isabelle-utp/utp-main/blob/7abe4b02af634ee70503afc39fab60f46a8cf954/casestudies/Tokeneer/Tokeneer.thy\#L495}
\begin{align*}
  UserTokenOK &\defs \left(\exists t @
  \begin{array}{l}
    currentUserToken = \mv{goodT}(t) \land t \in CurrentToken \land \\
    (\exists c \in IDCert @ idCert(t) = c \land CertOK c) \land \\
    (\exists c \in PrivCert @ privCert(t) = c \land CertOK c) \land \\
    (\exists c \in IandACert @ iandACert(t) = c \land CertOK c)
  \end{array} \right) 
\end{align*}
It requires that $currentUserToken$ contains a token, which is current ($CurrentToken$), and has valid ID, privilege,
and I\&A certificates. The definitions of the omitted predicates can be found
elsewhere~\cite{TIS-FormalSpec}.

Assuming these preconditions hold, operation \textit{BioCheckRequired} updates the state to $\mv{waitingFinger}$ and the
display with an instruction to provide a fingerprint. \textit{ReadFingerOK} requires that the state is
$\mv{waitingFinger}$, and checks whether both finger and user token are present. If they are, then the state switches to
$\mv{gotFinger}$, and the display is updated to $\mv{wait}$. \textit{UnlockDoorOK} requires that the current state is
$\mv{waitingRemoveTokenSuccess}$, and the token has been removed. It unlocks the door, using the auxiliary operation
\textit{UnlockDoor}, returns the status to $\mv{quiescent}$, and updates the display. \textit{UnlockDoor} both unlocks
the latch, and also updates two timeout variables, $latchTimeout$ and $alarmTimeout$. The former is used to close the
latch after a certain period, and the latter to sound an alarm if the door is left open.

These operations act only on the TIS state space. During their execution, monitored variables can also change, to reflect
real-world updates. Mostly these changes are arbitrary, with the exception that time must increase monotonically. We
therefore promote the operations to \textit{SystemState} with the following schema. \isalink{https://github.com/isabelle-utp/utp-main/blob/7abe4b02af634ee70503afc39fab60f46a8cf954/casestudies/Tokeneer/Tokeneer.thy\#L1372}
$$UEC(Op) \defs \uframe{tis}{Op} \relsemi \uframe{rw}{\lns{mon}{now} \le \lns{mon}{now'} \land ctrl' = ctrl}$$
In Z, this functionality is provided by the schema \textit{UserEntryContext}~\cite{TIS-FormalSpec}, from which we derive
the name \textit{UEC}. It promotes $Op$ to act on $tis$, and composes this with a relational predicate that constrains
the real-world variables ($rw$). The behaviour of all monitored variables other than $now$ is arbitrary, and all
controlled variables are unchanged. This separation enables modular reasoning, since we can promote invariants of the
TIS to any real world context using Theorem~\ref{thm:utp-wp} and the following Hoare logic theorem.

\begin{theorem}[TIS Promotion] If $\hoaretriple{I}{P}{I}$ then $\hoaretriple{I_{\uparrow tis}}{UEC(P)}{I_{\uparrow tis}}$ \isalink{https://github.com/isabelle-utp/utp-main/blob/d0356202092c61f7109cb9b30b464df44ce2f299/casestudies/Tokeneer/Tokeneer.thy\#L1386}
\end{theorem}
\noindent This shows that if $I$ is an invariant of $P$, an operation of the TIS, then $I$ in the extended state space
is an invariant of promoted operation. For comparison with the original Z operation schemas, we give the $ReadUserToken$
schema below:

\begin{schema}{ReadUserToken}
        UserEntryContext
\also
        \Xi UserToken
\\	\Xi DoorLatchAlarm
\\      \Xi Stats
\also
        AddElementsToLog
\where
        enclaveStatus \in \{ enclaveQuiescent,
        waitingRemoveAdminTokenFail \}
\also
	status = quiescent
\\	userTokenPresence = present
\also
	currentDisplay' = wait
\\	status' = gotUserToken
\end{schema}
\noindent We have employed a pattern for conversion from Z to GCL. Every conditional predicate, for example
$status = quiescent$, becomes a guard in Definition~\ref{def:tis-ops}. Every primed variable equation, such as
$status' = gotUserToken$, becomes an assignment, and all of the resulting commands are sequentially
composed. Nevertheless, we preserve the non-determinism of the original model, and these assignments are equivalent to
primed variable equations since in UTP we denote an assignment as follows:
$$(x := e) \defs x' = e \land y' = y \land \cdots \land z' = z$$
Consequently, the operations could be expressed as relational expressions. Using $UEC$ we promote each operation, for
example $TISReadToken \defs UEC(ReadToken)$, to achieve the same effect as including $UserEntryContext$.

In Z, invariants of the state are imposed through the inclusion of state schemas, such as $UserToken$. Here, we do not
impose these but we will prove that each operation preserves each invariant in Section~\ref{sec:formal-verify}. This
sometimes requires that we add extra assignments to satisfy the invariant, as we illustrate below.

We next define some of the key admin operations, which are necessary to prove the security properties.
\begin{definition}[Admin Operations] \label{def:admin-ops} \isalink{https://github.com/isabelle-utp/utp-main/blob/7abe4b02af634ee70503afc39fab60f46a8cf954/casestudies/Tokeneer/Tokeneer.thy\#L2429}
\begin{align*}
  OverrideDoorLockOK &\defs
  \begin{array}{l}
    \left(\begin{array}{l}
      enclaveStatus = \mv{waitingStartAdminOp} \\
      \land adminTokenPresence = \mv{present} \\
      \land currentAdminOp = Some(\mv{overrideLock})
    \end{array}\right) \\
    \longrightarrow
    \begin{array}{l}
      screenMsg := \mv{requestAdminOp} \relsemi currentDisplay := \mv{doorUnlocked} \relsemi \\
      enclaveStatus = \mv{enclaveQuiescent} \relsemi UnlockDoor \relsemi \\
      currentAdminOp := None
    \end{array}
  \end{array} \\
  FinishUpdateConfigOK &\defs
  \begin{array}{l}
  \left(\begin{array}{l}
    enclaveStatus = \mv{waitingFinishAdminOp} \\
    \land adminTokenPresence = \mv{present} \\     
    \land currentAdminOp = Some(\mv{updateConfigData}) \\
    \land floppyPresence = \mv{present} \land currentFloppy \in range(configFile) \\
    \land ValidConfig(\textit{configFileOf}(currentFloppy))
  \end{array}\right) \\
    \longrightarrow 
    \begin{array}{l}
      config := \textit{configFileOf}(currentFloppy) \relsemi \\
      screenMsg := \mv{requestAdminOp} \relsemi \\
      screenConfig := displayConfigData(config) \relsemi \\
      enclaveStatus = \mv{enclaveQuiescent} \relsemi currentAdminOp := None
    \end{array}
  \end{array} \\
  AdminLogout &\defs
  \begin{array}{l}
  rolePresent \neq None \\
    \longrightarrow
    \begin{array}{l}
      rolePresent := None \relsemi currentAdminOp := None \relsemi availableOps := \{\}
    \end{array}
  \end{array} \\
  ShutdownOK &\defs
  \begin{array}{l}
  \left(\begin{array}{l}
    enclaveStatus = \mv{waitingStartAdminOp} \\
    \land currentAdminOp = Some(\mv{shutdownOp}) \\
    \land currentDoor = \mv{closed}
  \end{array}\right) \\
    \longrightarrow
    \begin{array}{l}
      LockDoor \relsemi AdminLogout \relsemi screenMsg = \mv{clear} \relsemi \\
      enclaveStatus := \mv{shutdown} \relsemi currentDisplay := \mv{blank}
    \end{array}
  \end{array}
\end{align*}
\end{definition}
\noindent $OverrideDoorLockOK$ allows the door to be unlocked when an Admin has already logged in who can execute the
$\mv{overrideLock}$ command, that is, an Admin with the role $\mv{guard}$. If the enclave is awaiting an Admin command,
an Admin token is present, and the Admin gives the $\mv{overrideLock}$ command, then the door is unlocked and the
enclave returns to awaiting another Admin command. $FinishUpdateConfigOK$ is the second part of a two stage process for
updating the configuration file. The first step checks whether a configuration file floppy has been inserted. In this
second step, if the command \mv{updateConfigData} has been selected and a valid floppy has been inserted, then the
config is updated, displayed on the screen, and the enclave again returns to the main menu. Finally, $ShutdownOK$ is
used to shutdown the TIS. If an Admin is logged in, selects the $\mv{shutdownOp}$ command, and the door is closed, then
the operation locks the door, logs the Admin out, blanks the screen and display, and sets the status to
$\mv{shutdown}$. The auxiliary operation $AdminLogout$ has one more assignment than the corresponding Z
schema~\cite[page~41]{TIS-FormalSpec}, $availableOps := \{\}$, to ensure that the $Admin$ invariants in
Definition~\ref{def:admin-invar} are satisfied. In Z, this is implicit because the invariants are enforced at each
stage.

The overall behaviour of the TIS operations is given below:
\begin{align*}
 TISUserEntryOp &\defs \left(
 \begin{array}{l}
   TISReadUserToken \intchoice TISValidateUserToken \\
   \intchoice TISReadFinger \intchoice TISValidateFinger \\
   \intchoice TISWriteUserToken \intchoice TISValidateEntry \\
   \intchoice TISUnlockDoor \intchoice TISCompleteFailedAccess
 \end{array}\right) \\
 TISAdminOp &\defs \left(
 \begin{array}{l}
  TISOverrideDoorLockOp \intchoice TISShutdownOp \\
  \intchoice TISUpdateConfigDataOp \intchoice TISArchiveLogOp
 \end{array}\right) \\
 TISOp &\defs \left(
 \begin{array}{l}  
   TISEnrolOp \intchoice TISUserEntryOp \intchoice TISAdminLogon \intchoice TISStartAdminOp \\
   \intchoice TISAdminOp \intchoice TISAdminLogout \intchoice TISIdle
 \end{array}\right)
\end{align*}
We omit several operations, though these have all been mechanised. In each iteration of the state machine, we
non-deterministically select an enabled operation and execute it. We also update the controlled variables, which is done
by composition with the following relational update operation.
\begin{align*}
  TISUpdate \defs~& \uframe{rw}{\lns{mon}{now} \le \lns{mon}{now'}} \relsemi \lns{rw}{\lns{ctrl}{latch}} := \lns{tis}{currentLatch} \relsemi  \\
            & \lns{rw}{\lns{ctrl}{display}} := \lns{tis}{currentDisplay}
\end{align*}
This allows time to advance, allows other monitored variables to change, and copies the digital state of the latch and
display to the corresponding controlled variables. The system transitions are described by
$TISOp \relsemi TISUpdate$.

\subsection{Formal Verification}
\label{sec:formal-verify}

In this section, we verify three SFRs of the formal model using Isabelle/UTP. We first formalise the TIS state
invariants necessary to prove the SFRs\footnote{We adopt a different order for the invariants than our mechanisation,
  for the sake of presentation.}:
\begin{definition}[TIS State Invariants Selection] \isalink{https://github.com/isabelle-utp/utp-main/blob/7abe4b02af634ee70503afc39fab60f46a8cf954/casestudies/Tokeneer/Tokeneer.thy\#L527}
  \begin{align*}
    Inv_1 &\defs \left(
      \begin{array}{l}
      status \in 
            \left\{\begin{array}{l}
                     \mv{gotFinger}, \mv{waitingFinger}, \mv{waitingUpdateToken} \\
                     \mv{waitingEntry}, \mv{waitingUpdateTokenSuccess}
                   \end{array}\right\} \\
      \implies (UserTokenWithOKAuthCert \lor UserTokenOK)
      \end{array} \right) \\ %
    Inv_2 &\defs \left(
      \begin{array}{l}
      status \in \{\mv{waitingEntry}, \mv{waitingRemoveTokenSuccess}\} \\ 
        \implies (UserTokenWithOKAuthCert \lor FingerOK)
      \end{array} \right) \\ %
    Inv_3 &\defs (rolePresent \neq None \implies AdminTokenOK) \\ %
    Inv_4 &\defs \left(
       \begin{array}{l}
       currentAdminOp \in \{Some(\mv{shutdownOp}), Some(\mv{overrideLock})\} \\
       \implies ownName \neq None
       \end{array} \right) \\ %
    Inv_5 &\defs \left(
       \begin{array}{l}
       adminTokenPresent = present \land rolePresent \neq None \\       
       \implies rolePresent = Some(role(authCert(ofGoodT(currentAdminToken))))
       \end{array} \right) \\ %
    \textit{TIS-inv} &\defs \textit{TIS-wf} \land  Inv_1 \land Inv_2 \land Inv_3 \land Inv_4 \land Inv_5 \cdots 
  \end{align*}
\end{definition}
\noindent $Inv_1$ states that whenever the TIS is in a state beyond $\mv{gotUserToken}$, then either a valid
authorisation certificate is present, or else the user token is valid. It corresponds to the first invariant in the
$IDStation$ schema~\cite[page~26]{TIS-FormalSpec}. However, we need to add an extra state, $\mv{updateTokenSuccess}$ and
strengthen the consequent. The consequent originally only requires that there is a token with a valid authorisation
certificate, which may not be the case if a fingerprint has not yet been taken. $Inv_2$ states that whenever in state
$\mv{waitingEntry}$ or $\mv{waitingRemoveTokenSuccess}$, then either an authorisation certificate or a valid fingerprint
is present. $Inv_2$ is not present at all in~\cite{TIS-FormalSpec}, but we found it necessary to satisfy SFR1,
specifically to ensure that a valid fingerprint is present. That certain invariants are missing, or too weak, is
acknowledged in the TIS Security Properties document~\cite[page~11]{TIS-SecurityProperties}, but this does not
invalidate the functional specification; it just makes it tricky to formally verify the SFRs.

\begin{figure}
  \begin{center}
    \includegraphics[width=.8\linewidth]{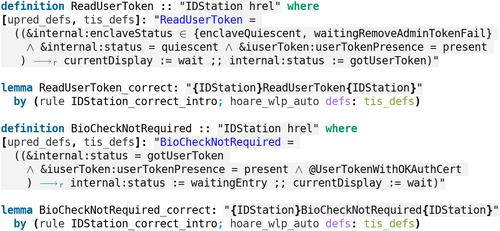}
  \end{center}
  
  \caption{Verification of Tokeneer Invariants in Isabelle/UTP}
  \label{fig:tok-inv-isa}
\end{figure}

$Inv_3$ states that whenever an Admin role is present, this means that a valid Admin token is also present
($AdminTokenOK$), so that it is not necessary to explicitly check this in each Admin operation. Similar to $Inv_1$, it is
equivalent to the second $IDStation$ schema invariant~\cite[page~26]{TIS-FormalSpec}, but we again needed to strengthen
the consequent. $Inv_4$ states that if an Admin operation $\mv{shutdownOp}$ or $\mv{overrideLock}$ is selected, then the
TIS must have an assigned name (also present in the key store), and hence it must already be enrolled. Finally, $Inv_5$
states that if an Admin token and Admin role are both present, then the role must match with the one contained on the
admin token. This invariant does not seem to be present at all in \cite{TIS-FormalSpec}, but we believe it is certainly
necessary to prove SFR1. We elide the additional five invariants that deal with administrators, the alarm, and audit
data~\cite{TIS-FormalSpec}.

As before, and differently to \cite{TIS-FormalSpec}, which imposes the invariants by construction, we prove that each
operation preserves the invariants using Hoare logic, similar to~\cite{Rivera2016-Undertakingtokeneerchallenge}:

\begin{theorem}[TIS Operation Invariants] \label{thm:tis-inv} \isalink{https://github.com/isabelle-utp/utp-main/blob/7abe4b02af634ee70503afc39fab60f46a8cf954/casestudies/Tokeneer/Tokeneer.thy\#L1807}
  \begin{itemize}
  \item $\hoaretriple{\textit{TIS-inv}}{TISUserEntryOp}{\textit{TIS-inv}}$
  \item $\hoaretriple{\textit{TIS-inv}}{TISAdminOp}{\textit{TIS-inv}}$
  \end{itemize}
\end{theorem}

\noindent This theorem shows that the user entry and admin operations never violate the well-formedness properties and
ten state invariants. We can therefore assume that they hold to satisfy any requirements. The proof involves discharging
verification conditions for a total of 32 operations in Isabelle/UTP, a process that is automated using our proof
tactics \inlineisar+hoare_auto+~\cite{Foster2020-IsabelleUTP} and \inlineisar+hoare_wlp_auto+. We illustrate this in
Figure~\ref{fig:tok-inv-isa} for two of the defined operations. We follow the mathematical notation for GCL as much as
possible. Each proof first applies an introduction rule, \inlineisar+IDStation_correct_intro+, that splits the goal into
the well-formedness and behavioural invariants. Then, \inlineisar+hoare_wlp_auto+ is applied to each resulting
goal. This high-level automation means that proofs can be adapted for small changes to the operations with minimal
intervention.

We use this fact to assure SFR1, which is formalised by the formula FSFR1, that characterises the conditions under which
the latch will become unlocked having been previously locked. We can determine these states by application of the
weakest precondition calculus~\cite{Dijkstra75}, which mirrors the (informal) Z schema domain calculations in~\cite[page
5]{TIS-SecurityProperties}. Specifically, we characterise the weakest precondition under which execution of
\textit{TISOp} followed by \textit{TISUpdate} leads to a state satisfying $\lns{rw}{\lns{ctrl}{latch}} =
\mv{unlocked}$. We formalise this in the definition below.

\begin{definition}[Formalisation of SFR1] \label{def:fsfr1} \isalink{https://github.com/isabelle-utp/utp-main/blob/7abe4b02af634ee70503afc39fab60f46a8cf954/casestudies/Tokeneer/Tokeneer.thy\#L2887}
  \begin{align*}
    \textit{AdminTokenGuardOK} & \defs \left(
      \begin{array}{l}
        \exists t \in TokenWithValidAuth @ currentAdminToken = \mv{goodT}(t)\\
        \land (\exists c \in AuthCert @ authCert(t) = Some(c) \land role(c) = guard)
      \end{array} \right) \\
    \textit{FSFR1} & \defs \left(
    \begin{array}{l}
    \left(\begin{array}{l}\textit{TIS-inv} \land \lns{tis}{currentLatch} = \mv{locked} \\
      \land (TISOp \relsemi TISUpdate)\mathop{\ckey{wp}\,}(\lns{rw}{\lns{ctrl}{latch}} = \mv{unlocked})
      \end{array} \right) \\
      \quad \implies
      \left(\begin{array}{l}
              (\textit{UserTokenOK} \land \textit{FingerOK}) \lor \textit{UserTokenWithOKAuthCert} \\
              \lor \textit{AdminTokenGuardOK}
            \end{array}\right)
    \end{array} \right)
  \end{align*}
\end{definition}
\noindent We first state the unlocking precondition for $TISOp$ using $\ckey{wp}$ calculus. Then, we conjoin the
\ckey{wp} formula with $\lns{tis}{currentLatch} = \mv{locked}$ to capture behaviours when the latch was initially
locked. The only operation that unlocks the door for users is $UnlockDoorOK$, and for admins it is
$OverrideDoorLockOK$. As a result, we can calculate the following unlocking preconditions.
\begin{theorem}[Unlocking Preconditions] \isalink{https://github.com/isabelle-utp/utp-main/blob/7abe4b02af634ee70503afc39fab60f46a8cf954/casestudies/Tokeneer/Tokeneer.thy\#L2894}
$$\begin{array}{l}
    ((TISUserOp \relsemi TISUpdate)\mathop{\ckey{wp}\,}(\lns{rw}{\lns{ctrl}{latch}} = \mv{unlocked}))[\mv{locked}/\lns{tis}{currentLatch}] \\[1ex]
    \qquad = (status = \mv{waitingRemoveTokenSuccess} \land userTokenPresence = \mv{absent}) \\[3ex]
    ((TISAdminOp \relsemi TISUpdate)\mathop{\ckey{wp}\,}(\lns{rw}{\lns{ctrl}{latch}} = \mv{unlocked}))[\mv{locked}/\lns{tis}{currentLatch}] \\[1ex]
    \qquad = \left(
    \begin{array}{l}
      enclaveStatus = \mv{waitingStartAdminOp} \land adminTokenPresence = \mv{present} \\
      \land currentAdminOp = Some(\mv{overrideLock}) \land rolePresent \neq None \land currentAdminOp \neq None
    \end{array}\right)
  \end{array}$$
\end{theorem}
\noindent The first equation shows that precondition for a user unlock is that access is permitted and the token has
been removed. The second equation shows that the precondition for an Admin unlock is that the TIS is waiting for an
Admin command, an Admin token is present, and the selected command is $\mv{overrideLock}$. From these equations we can
calculate the unlocking precondition of $TISOp$ itself, which is the disjunction of the two preconditions above. We can
then conjoin this with $\textit{TIS-Inv}$, since we know it holds in any state. We show that this composite precondition
implies that either a valid user token and fingerprint were present (using $Inv_2$) or a valid authorisation
certificate, or else an Admin is present (using $Inv_5$), and we can use the well-formedness invariant $Admin$ to show
that this Admin must have the $\mv{guard}$ role. Consequently, $FSFR1$ can indeed be verified.

\begin{theorem}[$FSFR1$ is provable] \label{thm:fsfr1} \isalink{https://github.com/isabelle-utp/utp-main/blob/7abe4b02af634ee70503afc39fab60f46a8cf954/casestudies/Tokeneer/Tokeneer.thy\#L2921}
\end{theorem}
\begin{proof}
  By application of weakest precondition and relational calculus.
\end{proof}

\noindent Proof of $SFR2$ can likely be achieved in a similar way to $SFR1$, but more complex additional invariants are
required that depend on time, which we have not been able to formalise for this case study (see \S\ref{sec:disc}).

Next, we consider $SFR3$, which requires that an alarm is raised if the door is left open. This property can be proved
more straightforwardly, since it is essentially a property of the well-formedness invariant in
$DoorLatchAlarm$~\cite[page~23]{TIS-FormalSpec}.
\begin{definition}[Formalisation of SFR3] $ $ \isalink{https://github.com/isabelle-utp/utp-main/blob/7abe4b02af634ee70503afc39fab60f46a8cf954/casestudies/Tokeneer/Tokeneer.thy\#L2976} %
  $$
  FSFR3 \defs \left(\begin{array}{l}
    IDStation \land
    currentLatch = lock \land currentDoor = open \land currentTime \ge alarmTimeout \\
    \qquad \implies doorAlarm = alarming
  \end{array}\right)
  $$
\end{definition}
This states that if the invariants hold, the latch is locked, the door is open, and the time has advanced beyond the
alarm timeout, then the door alarm is sounding. By Theorem~\ref{thm:tis-inv}, $DoorLatchAlarm$ always holds and
therefore FSFR3 can be verified.

\begin{theorem}[$FSFR3$ is provable] \label{thm:fsfr3} \isalink{https://github.com/isabelle-utp/utp-main/blob/7abe4b02af634ee70503afc39fab60f46a8cf954/casestudies/Tokeneer/Tokeneer.thy\#L2983}
\end{theorem}

Finally, we consider SFR6, which requires that the configuration and floppy can change only when an admin is logged
on. In order to verify this, we need to reason about the variables a given operation can modify. In Isabelle/UTP, we can
answer such framing questions using lenses~\cite{Foster07,Foster2020-IsabelleUTP}. We define a novel modification
predicate, $P \nmods a$, which states that relation $P$ does not modify any of the variables captured by $a$. This is
equivalent to stating that $P$ is a fixed point of the function $M(X) \defs (a' = a \land X)$. We prove the following
modification laws for this predicate.

\begin{theorem}[Modification Predicate] \label{thm:nmods} $ $ \isalink{https://github.com/isabelle-utp/utp-main/blob/7abe4b02af634ee70503afc39fab60f46a8cf954/utp/utp_rel.thy\#L704} %
  \vspace{1ex}
  
  \begin{tabular}{ccccccc}
    \AxiomC{---\vphantom{$P$}}
    \UnaryInfC{$\ckey{skip} \nmods x$}
    \DisplayProof
    &
    \AxiomC{---\vphantom{$P$}}
    \UnaryInfC{$\ckey{abort} \nmods x$}
    \DisplayProof
    &
    \AxiomC{$P \nmods x$}
    \AxiomC{$Q \nmods x$}
    \BinaryInfC{$(P \relsemi Q) \nmods x$}
    \DisplayProof
    \\[3ex]
    \AxiomC{$P \nmods x$}
    \AxiomC{$Q \nmods x$}
    \BinaryInfC{$(P \intchoice Q) \nmods x$}
    \DisplayProof
    &
    \AxiomC{$x \lindep y$}
    \UnaryInfC{$(y := v) \nmods x$}
    \DisplayProof
    &  
    \AxiomC{$P \nmods x$}
    \UnaryInfC{$(b \longrightarrow P) \nmods x$}
    \DisplayProof
    \\[3ex]
    \AxiomC{$x \notin a$}
    \UnaryInfC{$\uframe{a}{P} \nmods x$}
    \DisplayProof
    &
    \AxiomC{$P \nmods x$}
    \UnaryInfC{$\uframe{a}{P} \nmods a{:}x$}
    \DisplayProof
  \end{tabular}
\end{theorem}
\noindent As expected, neither $\ckey{skip}$ nor $\ckey{abort}$ modify any variable $x$. Sequential composition,
$P \relsemi Q$, does not modify $x$ provided neither $P$ nor $Q$ does, and similarly for internal choice. Assignment to
$y$ does not modify $x$ provided that $x$ is independent of $y$ ($x \lindep y$), which effectively means that $y$ is not
part of $x$. A guarded command $b \longrightarrow P$ does not modify $x$ provided that $P$ also does not. For the frame
operator, $\uframe{a}{P}$, we identify two cases. If a variable $x$ is not in $a$, then clearly it is not
modified. Conversely, if $x$ is within the $a$ namespace, then it is necessary to check whether $P \nmods x$. Using
these laws we can automatically verify that a program does not modify certain variables, and so formalise $SFR6$ as follows.
\begin{definition} $FSFR6 \defs (adminTokenPresence = \mv{absent} \longrightarrow TISOp) \nmods \{config, floppy\}$ \isalink{https://github.com/isabelle-utp/utp-main/blob/7abe4b02af634ee70503afc39fab60f46a8cf954/casestudies/Tokeneer/Tokeneer.thy\#L3072}
\end{definition}
\noindent If we assume that there is not an admin token, then this means that $TISOp$ cannot modify either $config$ or
$floppy$. For the verification, we can distribute the absence precondition throughout the operations using the law
$b \longrightarrow (P \intchoice Q) = (b \longrightarrow P) \intchoice (b \longrightarrow Q)$. One admin operation can
modify $config$, namely $FinishUpdateConfigOK$ in Definition~\ref{def:admin-ops}. If we prefix this operation with
$adminTokenPresence = \mv{absent}$, we obtain the program $\ckey{abort}$ which does not modify $config$, since this
violates the second guard. We can also prove that for every other operation $P$, $P \nmods config$ deductively using
Theorem~\ref{thm:nmods}. Consequently, we can prove $FSFR6$.

\begin{theorem}[$FSFR6$ is provable] \isalink{https://github.com/isabelle-utp/utp-main/blob/7abe4b02af634ee70503afc39fab60f46a8cf954/casestudies/Tokeneer/Tokeneer.thy\#L3074} \label{thm:fsfr6}
\end{theorem}

We have now formalised and verified three of the SFRs. In the next section we place these in the context of an assurance
case.

%% file: sec/new-tokassure.tex
\begin{figure}
  \begin{center}
    \includegraphics[width=.95\linewidth]{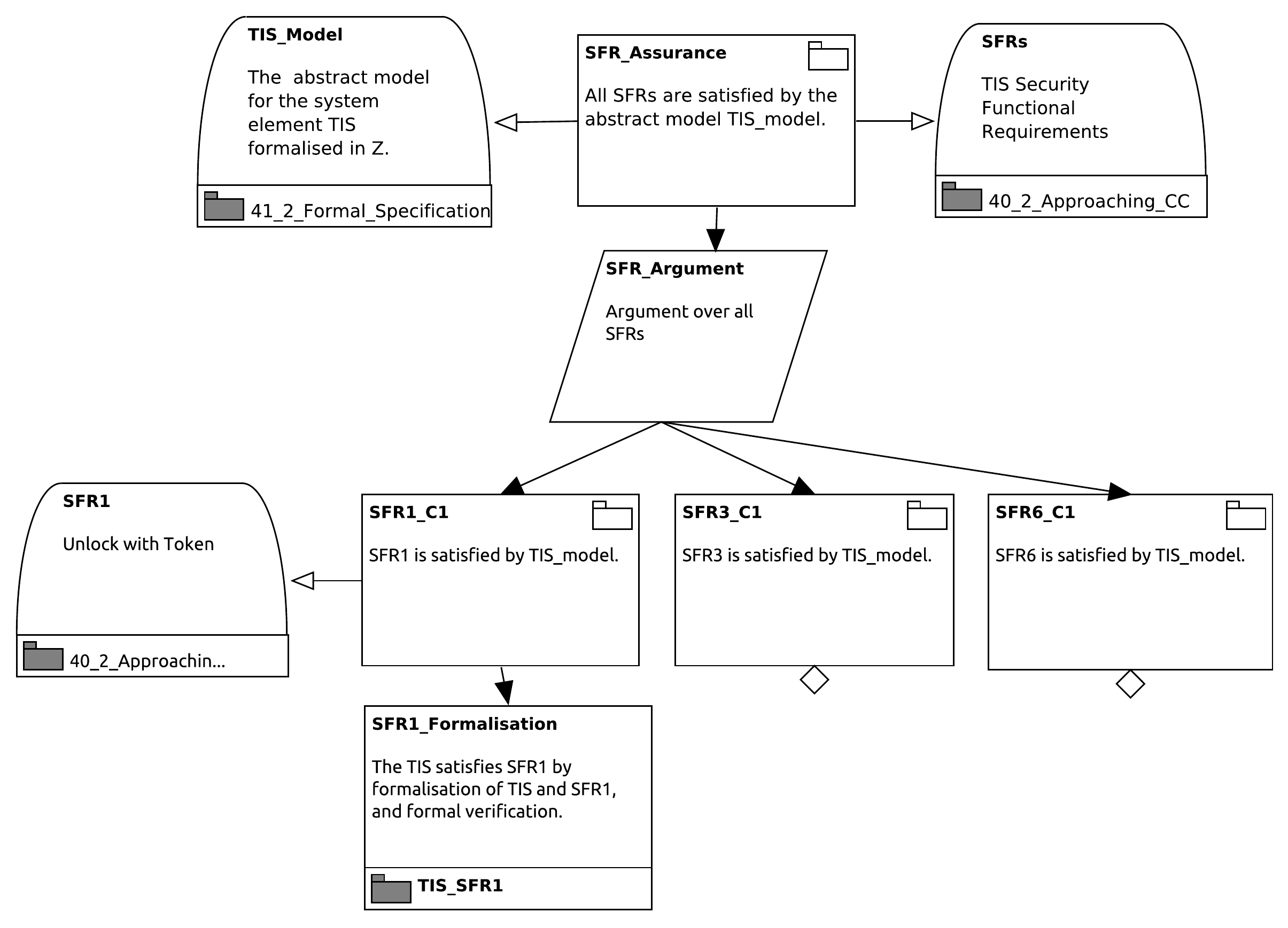}
  \end{center}

  \caption{Overall argument for SFR satisfaction ({\textsf{TIS\_SFRs}})}

  \label{fig:TIS_SFRs}
\end{figure}

In this section, we use ACME and \isacm to model the Tokeneer development process originally followed by Praxis, and
illustrated in Figure~\ref{fig:tis:proc}. A GSN diagram of the modular structure of the AC is shown in
Figure~\ref{fig:tis-modular}. The modules have names that correspond to the modelled artifact and a brief description.
Each module contains a mixture of lifecycle and certification artifacts, such as requirements and models, and GSN
arguments. The former were developed originally by Praxis and evaluated to comply with \ac{cc} EAL 5 in
the context of the certification process of TIS. While the certification artifacts were the contribution of Praxis, the
GSN AC argument modelling those artifacts is our contribution. Therefore, we complement the certification process of TIS
with a GSN model translated to \isacm, which aids the evaluation process of the certification artifacts by offering a
machine-checked argument structure with full traceability. Our work also provides guidelines on the use of modular GSN to
document the artifacts.

We focus on the module {\textsf{TIS\_SFRs}}, illustrated in Figure~\ref{fig:TIS_SFRs}. It encapsulates an argument for a
public claim that all SFRs are satisfied, which are defined in the module {\textsf{40\_2}}, by the TIS model, which is
defined in {\textsf{41\_2}}. We reference these artifacts with the use of \emph{away context} elements. For now, we
focus on the SFRs that we have formally verified, namely {\textsf{SFR1}}, {\textsf{SFR3}}, and {\textsf{SFR6}}. For
this, we use the Theorems~\ref{thm:tis-inv}, \ref{thm:fsfr1}, \ref{thm:fsfr3}, and \ref{thm:fsfr6} from
\S\ref{sec:model} as evidential artifacts. Satisfaction of {\textsf{SFR1}} is modelled by the claim {\textsf{SFR1\_C1}},
which uses {\textsf{SFR1}} as context. The claim is satisfied by formalisation, which is performed in another module
called {\textsf{TIS\_SFR1}}. Satisfaction of the other SFRs can be represented using the same pattern.

\begin{figure}
  \centering\includegraphics[width=.7\linewidth]{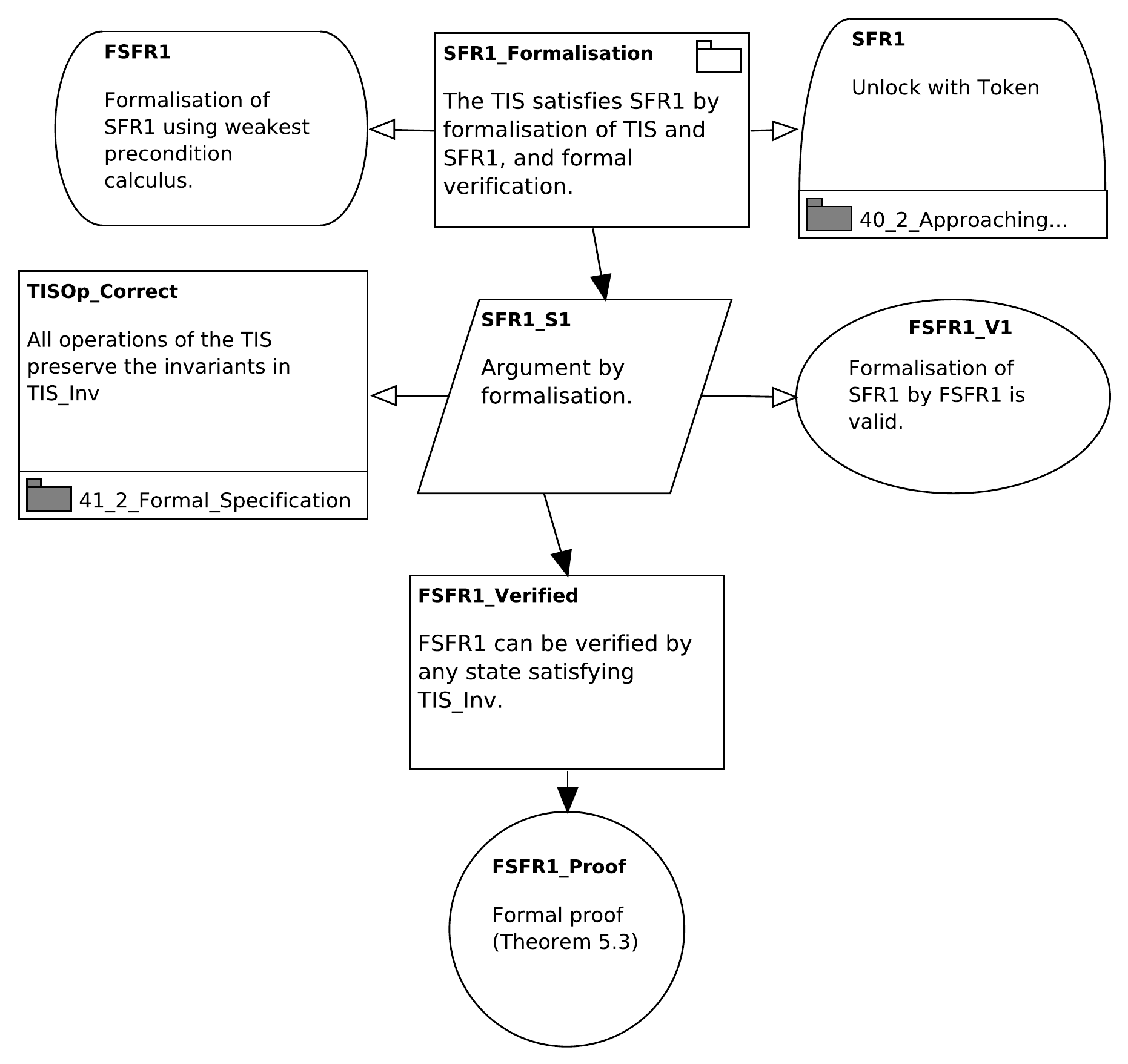}

  \caption{Satisfaction of SFR1 by formalisation}
  \label{fig:SFR1_Formalisation}
  
\end{figure}

The argument for {{SFR1}} is shown in Figure~\ref{fig:SFR1_Formalisation}. It uses the ``formalisation
pattern''~\cite{Denney2018}, which shows how results from a formal method can be used to provide evidence for claims to
satisfy a requirement \{R\} for a system element \{S\}. The strategy used to decompose the claim ``Requirement \{R\} is
met by \{S\}'' is contingent on the validation of both the formal specification of \{R\} and the formal model of \{S\}.
Consequently, the pattern breaks down the satisfaction of \{R\} into three claims stating that (1) the formal model of
\{S\} is validated, (2) the formalisation of \{R\} correctly characterises \{R\}, and (3) the formal model of \{S\}
satisfies the formalisation of \{R\}. The former two claims are usually satisfied by manual review, since they are not
amenable to formalisation, as the DO-178C formal methods supplement also makes clear:
\begin{center}
  \begin{minipage}{.9\linewidth}
  \textit{``Formal methods cannot show that derived requirements and the reason for their existence are correctly defined; this
  should be achieved by review.''}~\cite{DO-333}
\end{minipage}
\end{center}

In Figure~\ref{fig:SFR1_Formalisation}, we adapt Denney's pattern~\cite{Denney2018} as follows. We begin with the claim
{\textsf{SFR1\_Formalisation}}, which references both {\textsf{SFR1}}, with its natural language
description from module {\textsf{40\_1}}, and {\textsf{FSFR1}}, which is defined in this module and is
defined in Definition~\ref{def:fsfr1} from \S\ref{sec:model}. We then invoke an argumentation strategy,
{\textsf{SFR1\_S1}}, for formalisation. Instead of using a validation claim for the formalization of the
requirements, we use a justification element, {\textsf{FSFR1\_V1}}, which should be an explanation of how FSFR1
formalises SFR1. This is to preserve the well-formedness of the AC -- the ``requirement validation'' claims have a type
different from the ``requirement satisfaction'' claims.  An example of a ``requirement satisfaction'' claim is
{\textsf{SFR1\_Formalisation}}.

{\textsf{FSFR1}} also assumes that all the operations of the TIS preserve the system invariants, and so we record a link
to this proof in module {\textsf{41\_2}}, which corresponds to Theorem~\ref{thm:tis-inv}. The subclaim of
{\textsf{SFR1\_S1}} is {\textsf{FSFR1\_Verified}}, which is supported by the evidence {\textsf{FSFR1-Proof}}, which
refers to Theorem~\ref{thm:fsfr1}.

\begin{figure}
  \centering
  \includegraphics[width=.8\linewidth]{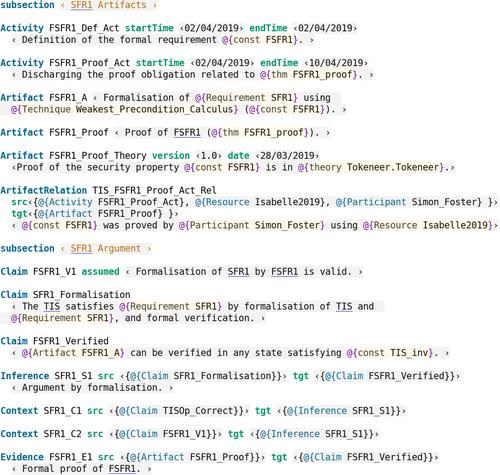}

  \caption{Argument and Artifacts for the FSFR1 Argument}
  \label{fig:SFR1-Isabelle}
\end{figure}

Figure~\ref{fig:SFR1-Isabelle} shows the IAL model of {\textsf{TIS\_SFR1}} that was manually translated and elaborated
from Figure~\ref{fig:SFR1_Formalisation}. In our translation, each of the modules in Figure~\ref{fig:tis-modular} is
assigned an Isabelle theory with the corresponding artifacts. We represent both the artifacts and argumentation elements
necessary to assure satisfaction of {\textsf{SFR1}}. Each command has an optional descriptive text, enclosed in quotes
\inlineisar+\<open>...\<close>+ that can integrate hyperlinks to both formal artifacts, such as theorems and proof, and
structured assurance artifacts, such as model elements generated by IAL. Since the checks done by IAL are successful, no
errors are issued in Figure~\ref{fig:SFR1-Isabelle}, which in particular indicates that every referenced artifact exists
and is correctly typed.

SACM provides several additional concepts for representing lifecycle artifacts, and we utilise them here. We record two
activities, {\textsf{FSFR1\_Def\_Act}} and {\textsf{FSFR1\_Proof\_Act}}, which represent the activities in the
development workflow for defining the formal requirement and discharging the proof obligations. Both have a
\textcolor{OliveGreen}{\inlineisar+startTime+} and \textcolor{OliveGreen}{\inlineisar+endTime+}
associated. {\textsf{FSFR1}} is represented by the artifact {\textsf{FSFR1\_A}}, which links to the IAL requirement
\textsf{SFR1}, which contains the natural language description of the requirement SFR1 from the Tokeneer documentation,
using the \inlineisar+Requirement+ antiquotation, and the technique
\textsf{Weakest\_Precondition\_Calculus}~\cite{Dijkstra75}. {\textsf{FSFR1\_A}} also contains a link to the
corresponding formal Isabelle constant via the antiquotation \inlineisar+@{const FSFR1}+.

We record a link to the proof of \textsf{FSFR1} in the artifact \textsf{FSFR1\_Proof}, and the Isabelle theory where
this resides in \textsf{FSFR1\_Proof\_Theory}. Finally, we create an artifact relation that gives the provenance for the
proof of \textsf{FSFR1}. This proof was performed by the participant \textsf{Simon\_Foster}, during the activity
\textsf{FSFR1\_Proof\_Act}, using the theorem prover \textsf{Isabelle2019}. In this way we record precisely how and when
a particular assurance artifact was created.

With the artifacts and their provenance defined, we move on to the argumentation. We first create the key claims using
the \textcolor{Blue}{\inlineisar+Claim+} command, which variously reference the artifacts previously defined. Claim
\textsf{FSFR1\_V1} is marked as \textcolor{OliveGreen}{\inlineisar+assumed+}, since this is the validation claim that
must be satisfied elsewhere by review.  The strategy {\textsf{SFR1\_S1}} from Figure~\ref{fig:SFR1_Formalisation}, is
modelled by \inlineisar+SFR1_S1+ in Figure~\ref{fig:SFR1-Isabelle}. \inlineisar+SFR1_S1+ is created using the command
\textcolor{Blue}{\inlineisar+Inference+}, which uses antiquotations to refer to the premise claims
\inlineisar+SFR1_Formalisation+, \inlineisar+TISOp_Correct+, and \inlineisar+FSFR1_V1+, that is, the source
\textcolor{OliveGreen}{\inlineisar+src+}, and the conclusion claim \inlineisar+FSFR1_Verified+, that is the target
\textcolor{OliveGreen}{\inlineisar+tgt+}.

We use the \textcolor{Blue}{\inlineisar+Context+} command to model the two contextual relations in
Figure~\ref{fig:SFR1_Formalisation}. \textsf{SFR1\_C1} presents the external claim \textsf{TISOp\_Correct} as context,
which refers to the invariant proof (Theorem~\ref{thm:tis-inv}), and \textsf{SFR1\_C2} presents the assumed validation
claim as context. Finally, we model the relationships from Figure~\ref{fig:SFR1_Formalisation} that link
{\textsf{FSFR1\_Verified}} to {\textsf{FSFR1\_Proof}}. This is done in Figure~\ref{fig:SFR1-Isabelle} by
\inlineisar+FSFR1_E1+, which is created using the command \textcolor{Blue}{\inlineisar+Evidence+}. It supports the claim
\inlineisar+FSFR1_Verified+ with the artifact \inlineisar+FSFR1_Proof+.

We have shown how \isacm enables the integration of formal development with assurance argumentation, documenting how the
evidence collected establish the overall security claims. In the next two sections we survey related work and discuss the
findings of our case study.

%% file: sec/related.tex
In this section, we discuss previous efforts in the verification of
Tokeneer as well as other approaches to the formalisation of assurance
cases and the integration of formal methods with assurance cases.

\paragraph{Comparison with Previous Efforts in the Verification of
Tokeneer}

Woodcock et al.~\cite{Woodcock2010-TokeneerExperiments} highlight
defects of the Tokeneer SPARK implementation, indicate undischarged
verification conditions, and perform robustness tests generated by the
Alloy SAT solver~\cite{Jackson2000-AlloyLightweightObject} from a
corresponding Alloy model.  Using De Bono's lateral thinking, these
test cases go beyond the anticipated operational envelope and
stimulate anomalous behaviours.  In shortening the feedback cycle for
verification and test engineers, theorem proving in form of the
proposed framework can help using this approach more intensively.

Abdelhalim et al.~\cite{Abdelhalim2010-FormalVerificationTokeneer}
model part of the Tokeneer specification using UML activity diagrams
translated to be checked for deadlock freedom by the CSP model checker
FDR.  Their formalisation assumes to be implemented on top of
asynchronous communication, modelled in CSP in terms of buffers for
each channel between UML components.  While our abstraction from such
communication aspects yields a simpler proof of the
\acp{sfr} in Section~\ref{props:TIS}, their deadlock checking at a
lower level can be useful for checking correctness of the
communication in an implementation of the UML model such as the SPARK
implementation mentioned in Figure~\ref{fig:tis:proc}.  Their UML
diagrams can lead to comparatively large specifications whereas our
formalisation stays compact thanks to the abstraction and reuse
mechanisms in Z schema and \iutp.

Rivera et al.~\cite{Rivera2016-Undertakingtokeneerchallenge} present
an Event-B model of the TIS, verify this model, generate Java code
from it using the Rodin tool, and test this code by JUnit tests
manually derived from the specification.  The tests validate the model
in addition to the Event-B invariants derived from the same
specification, and aim to detect errors in the Event-B model caused by
misunderstandings of the specification.  Using Rodin, the authors
state that they verify the \acp{sfr}~(Section~\ref{props:TIS}) using
Hoare triples.
Our work uses a similar abstract machine specification, but with
weakest precondition calculus as the main tool for verifying the
\acp{sfr}.
Beyond the replication of the Tokeneer case study, Rivera et
al.~\cite{Rivera2016-Undertakingtokeneerchallenge} deal with the
relationship between the model and the code via testing, whereas we
focus on the construction of certifiable assurance arguments from
formal model-based specifications.  Nevertheless, we believe
Isabelle's code generation features could be applied in a similar way.

\paragraph{Previous Work on Formal Assurance and Formalised Assurance Cases}

In concordance with Woodcock et
al.'s~\cite{Woodcock2010-TokeneerExperiments} observations, several
researchers have investigated ways of introducing formality into
assurance cases~\cite{Cruanes2013,Rushby2014,Denney2018,Diskin2018}.
We highlight some of these approaches below.

AdvoCATE is a powerful graphical tool for the construction of
GSN-based safety cases~\cite{Denney2018}.  It uses a formal foundation
called argument structures, which prescribe well-formedness checks for
the syntactic structure of (i.e.~the graph underlying) an \ac{ac}, and
allow instantiation of assurance case patterns.  Our work likewise
ensures well-formedness, but also allows the embedding of content with
formal semantics.  Denney and Pai's formalisation
pattern~\cite{Denney2018} is an inspiration for our work.
Our framework is to be used as an assurance backend, which complements
AdvoCATE with a deep integration of modelling and specification
formalisms.

Rushby~\cite{Rushby2014} illustrates how assurance arguments can be
formalized with modern verification systems such as Isabelle or PVS to
overcome some of the logical fallacies associated with informal
\acp{ac}.  Similarly, our framework allows reasoning using formal
logic, but additionally supports the combination of formal and
informal artifacts.  We drew inspiration from the work on the
Evidential Tool Bus~\cite{Cruanes2013}, which enables the combination
of evidence from several formal and semi-formal analysis tools.  In a
very similar way, Isabelle supports the integration of a variety of
formal analysis tools~\cite{Wenzel2007}.

Diskin et al.~\cite{Diskin2018} tackle the problem of hierarchical and
modular assurance by using a formal model~(in this case, a
compositional data-flow model) of the system to be assured as the basis
for generating evidence required for a particular
assurance claim.  Their framework is elaborate and practically 
relevant inasmuch as it integrates well with the practice of
model-based development.  The paradigm of our approach is similar to
theirs except that we use mechanised algebraic reasoning techniques,
provide computer-assistance for the proposed reasoning steps,
integrating informal assurance evidence.

Overall, we believe that our work is the first to put formal
verification effort into the wider context of structured assurance
argumentation, in our case, a machine-checked security case using
\isacm.  We have also recently applied our techniques to collision
avoidance for autonomous ground robots~\cite{Gleirscher2019-SEFM} and
an autonomous underwater vehicle~\cite{foster2020formal}; both of
which are more recent benchmark examples.

\section{Findings and Limitations of the Case Study}
\label{sec:disc}

Below, we summarise several observations and findings from our
investigation.

\paragraph{Evaluation of the Tokeneer Assurance Case}

Despite its age, we see Tokeneer as a highly relevant benchmark
specification, particularly since it is one of the grand challenges of
the ``Verified Software Initiative''~\cite{Woodcock2006}.  As we have
argued elsewhere~\cite{Gleirscher2018-NewOpportunitiesIntegrated},
such benchmarks allow us to conduct objective analyses of assurance
techniques to aid their transfer to other domains.  The issues
highlighted in \cite{Woodcock2010-TokeneerExperiments} are systematic
design problems that can be fixed by a change of the
benchmark~(e.g.~by a two-way biometric identification on both sides of
the enclave entrance).  However, this is out of scope of our work and
does not harm Tokeneer in its function as a benchmark.

During the translation from Z into \iutp's \ac{gcl} and the formalisation
of the \acp{sfr}, we identified some deficiencies in the way that the
security requirements were originally proven. In particular, as we
have previously mentioned, the developers acknowledge that there are
missing invariants necessary to support the proofs:
\begin{center}
\begin{minipage}{.9\linewidth}
\textit{We have not [added the invariants], as we believe it will add little to the
  assurance of correctness, and is very time consuming. At higher
  levels of the \ac{cc} assurance we would be required to carry out more
  formal proofs, in which case these modifications would be
  done.}~\cite[page~11]{TIS-SecurityProperties}
\end{minipage}
\end{center}
One of the reasons we can now do this is because automation of formal
proof has vastly improved since the development of
Tokeneer. Consequently, we can reach these higher assurance levels
with our mechanisation.

A further issue is that we could not prove \ac{sfr}2 while staying
faithful to its proposed formalisation in the benchmark
artifacts~\cite[page~6]{TIS-SecurityProperties}. This property states
that, at the point of unlocking the door, the time must be close to
being within the permitted entry period. Like \ac{sfr}1, it uses the
operation $TISOpThenUpdate$ as the target for the
verification. However, this operation does not allow $currentTime$,
which internally records the time, to advance. The advance of internal
time occurs only when polling $currentTime$ from the corresponding
monitored variables $now$, using $TISPoll$, which can advance
arbitrarily. Consequently, an invariant of time can be trivially
satisfied, because time is constant. A fix for the issue would require
us to reason about $TISPoll$, and thus a more substantial proof

\paragraph{Formal Design and Refinement}

As shown in Figure~\ref{fig:tis:proc} on page~\pageref{fig:tis:proc},
a complete assurance case of the \ac{tis} development would require
the coverage of all three refinement steps described in
\cite{TIS-SummaryRep}, the functional or abstract \emph{formal
  specification}, the more concrete \emph{formal design}, and the
\emph{SPARK implementation}.
The formal design is a data and operation refinement of the abstract
types used in the formal specification, replacing sets and functions
with data structures with operational semantics.
Such refinement proofs would require formal reasoning about the memory
models of the formal design and the SPARK implementation in \iutp.
This reasoning can be based on separation logic like, for example,
implemented in the Isabelle data refinement
library~\cite{Lammich2017-RefinementImperativeHOL}.

%% file: sec/concl.tex
We have presented \isacm, a framework for integrating formal proof into a unified and standardised form of assurance
cases and for their computer-assisted construction. We showed how SACM is embedded into Isabelle as an ontology, and
provided an interactive assurance language that guides its user in generating valid instances of this ontology.

We applied this framework to part of the Tokeneer security case, including the verification of three of the security
functional requirements, and embedded these results into a mechanised assurance argument. \isacm enforces the usage of
formal ontological links---a feature inherited from \idof---which establish and enrich traceability between the
assurance arguments, evidence of different provenance, and the assurance claims.  \isacm combines features from \ihol,
\idof, and \sacm in a way that allows integration of formal methods and \aclp{ac}~\cite{Gleirscher2019-SEFM}. In sum,
our work allows us to intertwine a heterogeneous formal development in Isabelle with an assurance case that puts the formal
results in context.

In future work, we will formalise the connection between ACME~\cite{Wei2019-SACM} and \isacm, which will make the
platform more accessible to safety practitioners. We are currently working on a prototype model-to-text transformation
from SACM to Isabelle to facilitate this, and an integration with Eclipse to allow feedback from Isabelle to be
propagated back to the diagram editors.
We will also consider the integration of \ac{ac} pattern execution~\cite{Denney2018}, to facilitate \ac{ac}
production. Moreover, to support more advanced safety analysis, we are exploring the use of \idof to develop an
ontology for safety concepts, such as hazards, risks, and control measures, following work by
Banham~\cite{Bahham2020-RiskOntology} and the Safety of Autonomous Systems Working Group~\cite{ASObjectives}, and also
ontologies for formal methods. Indeed, we envisage the development of a variety of ontologies that provide the necessary
terminology to formulate requirements, document developments, and otherwise aid communication.

We also plan complete the mechanisation of the TIS security case, including the overarching argument for how the formal
evidence can satisfy the requirements of \ac{cc}~\cite{CC2017-CommonCriteriaInformationPt1}. This will involve
mechanising the remaining operators, and tackling the last three security requirements, which will require derivation
and verification of additional invariants. We are also applying \isacm to develop an assurance case for an autonomous
underwater vehicle safety controller~\cite{foster2020formal}, which is being developed under the regime of the DO-178C
standard.

In parallel, we are developing our verification framework,
Isabelle/UTP~\cite{Foster2020-IsabelleUTP,Foster16a,Foster19a-IsabelleUTP} to support a variety of software engineering
notations.  We recently demonstrated formal verification facilities for a StateChart-like
notation~\cite{Foster18b,Foster17c}, and are also working towards tools to support hybrid dynamical
languages~\cite{Foster16b,Foster19c-HybridRelations,Foster2020-dL} like Modelica and MATLAB Simulink. Though these kinds
of models seem quite different to Tokeneer, in the UTP they all have a relational semantics and so the development of
proof facilities here feed into these works. 

Our long-term overarching goal is a comprehensive assurance framework supported by a variety of integrated \acp{fm}, in order to
support complex certification tasks for cyber-physical systems such as autonomous
robots~\cite{Gleirscher2018-NewOpportunitiesIntegrated,Gleirscher2019-SEFM,foster2020formal}.